\documentclass[11pt,reqno]{amsart}
\usepackage{amsfonts, amsmath, amssymb, color}
\usepackage{amsthm}
\usepackage{booktabs}
\usepackage{enumitem}
\usepackage{float}
\usepackage{hhline}
\usepackage{soul, color, xcolor}
\usepackage{lipsum}
\usepackage{lscape}
\usepackage{subfig}
\usepackage{textcomp}
\usepackage{tikz}
\usetikzlibrary{decorations.pathmorphing,shapes,arrows,positioning}
\usepackage{xcolor}
\usepackage{young}
\usepackage{youngtab}
\usepackage{ytableau}
\usepackage{rotating}
\usepackage{graphicx}
\usepackage[colorlinks=true, linkcolor=blue, citecolor=blue, urlcolor=blue]{hyperref} 
\usepackage{caption}

\hypersetup{colorlinks,linkcolor=blue,urlcolor=blue,citecolor=blue}

\usepackage{scalefnt}
\usetikzlibrary{calc}
\usepackage{a4wide}

\usepackage[numbers,sort&compress]{natbib}

\allowdisplaybreaks[4]
\theoremstyle{plain}
\newtheorem{thm}{Theorem}[section]
\newtheorem{lem}[thm]{Lemma}

\newtheorem{cor}[thm]{Corollary}

\theoremstyle{definition}

\newtheorem{exmp}[thm]{Example}

\numberwithin{equation}{section} \errorcontextlines=0

\makeatletter

\begin{document}
\title{Superior monogamy and polygamy relations and
estimates of concurrence}
\author{Yue Cao}
\address{College of Arts and Sciences, Guangzhou Maritime University, Guangzhou, Guangdong 510725, China}
\email{434406296@qq.com}
\author{Naihuan Jing*}
\address{Department of Mathematics, North Carolina State University, Raleigh, NC 27695, USA}
\email{jing@ncsu.edu}
\author{Kailash Misra}
\address{Department of Mathematics, North Carolina State University, Raleigh, NC 27695, USA}
\email{misra@ncsu.edu}
\author{Yiling Wang}
\address{Department of Mathematics, North Carolina State University, Raleigh, NC 27695, USA}
\email{ywang327@ncsu.edu}
\thanks{Supported in part by Simons Foundation grant nos. 36482 and MP-TSM-00002518}
\thanks{$*$Corresponding author: jing@ncsu.edu}

\begin{abstract}
 It is well known that any well-defined bipartite entanglement measure $\mathcal{E}$ obeys $\gamma$th-monogamy relations Eq. \eqref{e:chap1-ineq1} and assisted measure $\mathcal{E}_{a}$ obeys $\delta$th-polygamy relations Eq. \eqref{e:chap1-ineq2}. Recently, we presented a class of tighter parameterized monogamy relation for the $\alpha$th $(\alpha\geq\gamma)$ power 
 based on Eq. \eqref {e:chap1-ineq1}. This study provides a family of tighter lower (resp. upper) bounds of the monogamy (resp. polygamy) relations in a unified manner. In the first part of the paper, the following three basic problems are focused:
\begin{enumerate}[label= {\bf(\roman*)}]
\item tighter monogamy relation for the $\alpha$th ($0\leq \alpha\leq \gamma$) power of any bipartite entanglement measure $\mathcal{E}$ based on Eq. \eqref{e:chap1-ineq1};
\item tighter polygamy relation for the $\beta$th ($ \beta \geq \delta$) power of any bipartite assisted entanglement measure $\mathcal{E}_{a}$ based on Eq. \eqref{e:chap1-ineq2};
\item tighter polygamy relation for the $\omega$th ($0\leq \omega \leq \delta$) power of any bipartite assisted entanglement measure $\mathcal{E}_{a}$ based on Eq. \eqref{e:chap1-ineq2}.
\end{enumerate}

In the second part, using the tighter polygamy relation for the $\omega$th ($0\leq \omega \leq 2$) power of CoA, we obtain good
estimates or bounds for the $\omega$th ($0\leq \omega \leq 2$) power of concurrence for any $N$-qubit pure states $|\psi\rangle_{AB_{1}\cdots B_{N-1}}$ under the partition $AB_{1}$ and $B_{2}\cdots B_{N-1}$. Detailed examples are given to illustrate that our findings exhibit greater strength across all the region.

\end{abstract}

\maketitle

\section{\textbf{Introduction}}
As one of the essential resources in quantum communication and quantum information processing, quantum entanglement holds great significance \cite{CAF,HHHH,DFSC}. Unlike the classical correlations, a critical property of entanglement is that a quantum system sharing entanglement with one of the subsystems is not free to share entanglement with the rest of the remaining systems. This property is usually called monogamy \cite{ASI}, which characterizes the entanglement distribution in multipartite systems. The monogamy relation has important applications in quantum key distribution,
quantum communications \cite{B,CB,GR}, etc.

For a tripartite quantum state $\rho_{ABC}$, entanglement measure $\mathcal{E}$ is called monogamous if
$
\mathcal{E}(\rho_{A|BC})\geq \mathcal{E}(\rho_{AB})+\mathcal{E}(\rho_{AC}),
$
 where $\rho_{AB}$ and $\rho_{AC}$ are the reduced density matrices of $\rho_{ABC}$. In general, entanglement measure $\mathcal{E}$ violates this inequality, while $\mathcal{E}^{\alpha}$ satisfies the monogamy relation for some $\alpha>0.$
 Coffman et al \cite{CKW} first discovered this inequality for the squared concurrence $C^{2}$, and it was generalized to multipartite qubit systems by Osborne and Verstraete \cite{OV}.
Since then, monogamy has been studied for many different situations \cite{G1,JHC,JLLF,ZF1,KPSS, KDS,OF}.

The assisted entanglement is the dual concept of entanglement. As another entanglement constraint in multipartite systems, it has the property of being viewed as a dual form of monogamy, which is called polygamy. For a tripartite quantum state $\rho_{ABC}$, polygamy of entanglement can be described by
$
\mathcal{E}_{a}(\rho_{A|BC})\leq \mathcal{E}_{a}(\rho_{AB})+\mathcal{E}_{a}(\rho_{AC})
$
with a bipartite assisted entanglement $\mathcal{E}_{a}$.
Gour et al \cite{GMS} established the polygamy inequality by using the squared concurrence of assistance $C^{2}_{a}$, which was quickly generalized to multipartite qubit systems \cite{GBS}.
Generalized polygamy inequalities of multipartite entanglement of assistance are also proposed in \cite{K2}.  

The $\gamma$th-monogamy $(\gamma> 0)$ relation of the measure $\mathcal{E}$ for any $N$-qubit state $\rho_{AB_{1}\cdots B_{N-1}}$ is defined as
\cite[Thm. 1, Def. 1]{GG}
\begin{align}\label{e:chap1-ineq1}
\mathcal{E}^{\gamma}(\rho_{A|B_{1}\cdots B_{N-1}})\geq \sum_{i=1}^{N-1} \mathcal{E}^{\gamma}(\rho_{AB_{i}}).
\end{align}
where $\rho_{AB_{i}}=\operatorname{Tr}_{B_{1}\cdots B_{i-1}B_{i+1}\cdots B_{N-1}}(\rho_{A B_{1}\cdots B_{N-1}})$  is the reduced density matrix. The exponent $\gamma$ depends on the infimum of all indices satisfying monogamy
relation \eqref{e:chap1-ineq1} of measure $\mathcal{E}$ (eg. If $\mathcal{E}=C$, then $\gamma=2$).

 The $\delta$th-polygamy $(\delta> 0)$ relation of assisted entanglement measure $\mathcal{E}_{a}$ for any $N$-qubit state $\rho_{AB_{1}\cdots B_{N-1}}$ is described as
 \cite[Thm 1, Def. 1]{G}
\begin{align}\label{e:chap1-ineq2}
\mathcal{E}_{a}^{\delta}(\rho_{A|B_{1}\cdots B_{N-1}})\leq \sum_{i=1}^{N-1} \mathcal{E}_{a}^{\delta}(\rho_{AB_{i}}).
\end{align}
Here the exponent $\delta$ depends on the supremum of all indices satisfying polygamy
relation \eqref{e:chap1-ineq2} of assisted measure $\mathcal{E}_{a}$ (eg. If $\mathcal{E}_{a}=C_{a}$, then $\delta=2$).

It is worth looking for tighter monogamy and polygamy relations, which can provide a better characterization of the distribution of quantum correlations. Hence the research for tight monogamy and polygamy relations has also attracted widespread attention. One common method to study monogamy and polygamy relations is to bound the binomial function $(1+t)^{x}$ using various smart estimates \cite{ZF2,JF1,ZJZ, YCFW,TZJF,JFQ,ZLJM}. Recently, we presented a family of tighter weighted $\alpha$th-monogamy $(0\leq \alpha\leq\gamma)$  relations \cite{CJW} and tighter parameterized $\alpha$th-monogamy $(\alpha\geq\gamma)$ relations \cite{CJMW} based on Eq. \eqref {e:chap1-ineq1}.

In this study, we propose a new method about  the binomial function $(1+t)^{x}$ by parametric inequalities.
We give a family of tighter monogamy relation for the $\alpha$th ($0\leq \alpha\leq \gamma$) power of any bipartite measure $\mathcal{E}$ based on Eq. \eqref {e:chap1-ineq1}, as well as tighter polygamy relations for the $\beta$th ($ \beta \geq \delta$) power and $\omega$th ($0\leq \omega \leq \delta$) power of any bipartite assisted measure $\mathcal{E}_{a}$ based on Eq. \eqref {e:chap1-ineq2} in a unified manner.

Our study also enables us to estimate the entropy or
concurrence assisted, the second part of the paper will be devoted to give good estimates for the measure. One finds that
our bounds are significant better than some of the known bounds in the literature.

This paper is organized as follows. In Section \ref{s:monogamy} we give tighter monogamy relation for the $\alpha$th ($0\leq \alpha \leq \gamma$) power of any bipartite measure $\mathcal{E}$ based on the mathematical results from \cite{CJW}.
In Section \ref{s:polygamy1} we investigate tighter polygamy relation for the $\beta$th ($\beta \geq \delta$) power of any bipartite assisted measure $\mathcal{E}_{a}$ based on \cite{CJW}. In Section \ref{s:polygamy2} we first prepare the necessary mathematical tools to deal with the approximation and thus give tighter polygamy relation for the $\omega$th ($0\leq \omega \leq \delta$) power of assisted measure $\mathcal{E}_{a}$. Based on this, we obtain tighter lower and upper bounds of the $\omega$th ($0\leq \omega \leq 2$) power of concurrence for any $N$-qubit pure states $|\psi\rangle_{AB_{1}\cdots B_{N-1}}$ under the partition $AB_{1}$ and $B_{2}\cdots B_{N-1}$. Especially, we give three examples to illustrate why our new bounds are stronger than some of the recently found sharper bounds.

\section{\textbf{Tighter $\alpha$th $(0\leq\alpha \leq \gamma)$ power monogamy relations of entanglement measures }}\label{s:monogamy}

Let $\rho=\rho_{AB_{1}\cdots B_{N-1}}$ be an $N$-partite quantum state over the Hilbert space $\mathcal{H}_{A}\bigotimes $ $ \mathcal{H}_{B_{1}} $ $
\bigotimes $ $ \cdots$ $ \bigotimes\mathcal{H}_{B_{N-1}}$. If there is no confusion, we will simply write $\mathcal{E}_{(a)}(\rho_{AB_{i}})=\mathcal{E}_{(a)AB_{i}}$ and
$\mathcal{E}_{(a)}(\rho_{A|B_{1}B_{2}\cdots B_{N-1}})=\mathcal{E}_{(a)A|B_{1}B_{2}\cdots B_{N-1}}$ etc.

\bigskip
In order to obtain tighter monogamy relation for the $\alpha$th ($0\leq \alpha \leq \gamma$) power of entanglement measures $\mathcal{E}$, we first recall
the following lemma: 

\begin{lem} \cite{CJW} \label{l:chap2-lem1}
Let $a\geq 1$ be a real number. For $t\geq a\geq 1$ and $0 \leq x\leq 1$, we have that
\begin{equation}\label{e:chap2-ineq1}
\begin{aligned}
(1+t)^{x}&\geq\left(1+\frac{a}{s}\right)^{x-1}+\left(1+\frac{s}{a}\right)^{x-1}t^{x} \geq\left(1+a\right)^{x-1}+\left(1+\frac{1}{a}\right)^{x-1}t^{x} \\
&\geq 1+\frac{(1+a)^{x}-1}{a^{x}}t^{x} \geq 1+(2^{x}-1)t^{x}
\end{aligned}
\end{equation}
for any parameter $s \in [\frac{a}{t}, 1]$.
\end{lem}

\bigskip

\begin{lem}\label{l:chap2-lem2} Let $p_i$ be $N$ positive numbers such that $p_{i}\geq  p_{i+1}(i=1,\cdots,N-1)$, then one has that
\begin{equation}\label{e:chap2-ineq2}
\begin{aligned}
\left(\sum_{i=1}^{N}p_{i}\right)^{x} \geq \sum_{k=1}^{N-1}\left(1+\frac{k-1}{s}\right)^{x-1}\prod_{j=1}^{N-k}\left(1+\frac{s}{N-j}\right)^{x-1}p_{k}^{x}+\left(1+\frac{N-1}{s}\right)^{x-1}p_{N}^{x}
\end{aligned}
\end{equation}
for $0\leq x \leq 1$, where $r\leq s \leq 1$ and $r=max\left\{\frac{hp_{h+1}}{p_{1}+\cdots+p_{h}}|~~h=1,\cdots,N-1\right\}.$ 
\end{lem}
\begin{proof} We use induction on $N$. The case of $N = 1$ is clear. Assume Eq. \eqref{e:chap2-ineq2} holds for $<N$. For given $p_i$ it is clear that $p_{1}+p_{2}+\cdots+p_{N-1}\geq (N-1)p_{N}$. Using Lemma \ref{l:chap2-lem1} we have that
\begin{align*}
\left(\sum_{i=1}^{N}p_{i}\right)^{x}&=\left(p_{1}+p_{2}+\cdots+p_{N}\right)^{x}=p_{N}^{x}\left(1+\frac{p_{1}+p_{2}+\cdots+p_{N-1}}{p_{N}}\right)^{x}\\
&\geq p_{N}^{x}\left(1+\frac{N-1}{s}\right)^{x-1}+\left(1+\frac{s}{N-1}\right)^{x-1}\left(p_{1}+p_{2}+\cdots+p_{N-1}\right)^{x}\\
\end{align*}
where $\frac{(N-1)p_{N}}{p_{1}+\cdots+p_{N-1}}\leq s\leq 1.$
By the inductive hypothesis, the above is no less than the right-hand side (RHS) of Eq. \eqref{e:chap2-ineq2}.
\end{proof}

The following result is a direct consequence of
Lemma \ref{l:chap2-lem2}.

\begin{thm} \label{t:chap2-thm1}
Let $\mathcal{E}$ be a bipartite entanglement measure satisfying the $\gamma$th-monogamy \eqref{e:chap1-ineq1} and $\rho_{AB_{1}\cdots B_{N-1}}$ any $N$-qubit quantum state. Arrange $\{\mathcal{E}_{i}=\mathcal{E}_{AB_{i'}}|i=1,\cdots,N-1\}$ in descending order. If $\mathcal{E}^{\gamma}_{i}\geq \mathcal{E}^{\gamma}_{i+1}>0$ for $i=1,\cdots,N-2$, then
\begin{align*}
\mathcal{E}^{\alpha}_{A|B_{1}\cdots B_{N-1}} \geq \sum_{k=1}^{N-2}\left(1+\frac{k-1}{s}\right)^{\frac{\alpha}{\gamma}-1}\prod_{j=2}^{N-k}\left(1+\frac{s}{N-j}\right)^{\frac{\alpha}{\gamma}-1}\mathcal{E}_{k}^{\alpha}+\left(1+\frac{N-2}{s}\right)^{\frac{\alpha}{\gamma}-1}\mathcal{E}_{N-1 }^{\alpha},
\end{align*}
for $0\leq\alpha\leq \gamma$, where $q\leq s \leq 1$ and $q=max\left\{\frac{h\mathcal{E}^{\gamma}_{h+1}}{\mathcal{E}^{\gamma}_{1}+\cdots+\mathcal{E}^{\gamma}_{h}}|~~h=1,2,\cdots,N-2\right\}.$
\end{thm}

{\bf Comparison of the monogamy relations for entanglement measure $\mathcal{E}$}. By Theorem \ref{t:chap2-thm1} and Lemma \ref{l:chap2-lem1}, the following unified monogamy relations of $\alpha$th $(0\leq \alpha \leq \gamma)$ power of $\mathcal{E}$ hold.

\begin{align*}
\mathcal{E}^{\alpha}_{A|B_{1}\cdots B_{N-1}} & \geq \sum_{k=1}^{N-2}\left(1+\frac{k-1}{s}\right)^{\frac{\alpha}{\gamma}-1}\prod_{j=2}^{N-k}\left(1+\frac{s}{N-j}\right)^{\frac{\alpha}{\gamma}-1}\mathcal{E}_{k}^{\alpha}+\left(1+\frac{N-2}{s}\right)^{\frac{\alpha}{\gamma}-1}\mathcal{E}_{N-1 }^{\alpha}\\
&\geq \sum_{k=1}^{N-2}k^{\frac{\alpha}{\gamma}-1}\prod_{j=2}^{N-k}\left(1+\frac{1}{N-j}\right)^{\frac{\alpha}{\gamma}-1}\mathcal{E}_{k}^{\alpha}+\left(N-1\right)^{\frac{\alpha}{\gamma}-1}\mathcal{E}_{N-1 }^{\alpha}\\
&\geq \sum_{k=1}^{N-2}\prod_{j=2}^{N-k}\frac{(N-j+1)^{\frac{\alpha}{\gamma}}-1}{(N-j)^{\frac{\alpha}{\gamma}}}\mathcal{E}_{k}^{\alpha}+\mathcal{E}_{N-1 }^{\alpha}\\
&\geq \sum_{k=1}^{N-1}(2^{\frac{\alpha}{\gamma}}-1)^{N-1-k}\mathcal{E}_{k}^{\alpha}
\end{align*}
where $s, q$ are defined as in Theorem \ref{t:chap2-thm1}.

\bigskip
Now let's take the concurrence to demonstrate our bounds of the $\alpha$th $(0\leq \alpha \leq \gamma)$ power monogamy relations perform best among recent studies.

Recall that the concurrence of a pure
state $\rho_{AB} \in\mathcal{H}_A \otimes \mathcal{H}_B$ is defined in \cite{ RBCGM,U} by
\begin{equation}\label{e:chap2-ineq3}
C\left(|\psi\rangle_{A B}\right)=\sqrt{2\left[1-\operatorname{Tr}\left(\rho_A^2\right)\right]}=\sqrt{2\left[1-\operatorname{Tr}\left(\rho_B^2\right)\right]},
\end{equation}
where $\rho_A$ (resp. $\rho_B$) is the reduced density matrix by tracing over the subsystem $B$ (resp. $A$).

For a mixed state $\rho_{A B}$, the concurrence and concurrence of assistance (CoA) \cite{YS} are given by
 \begin{equation}\label{e:chap2-ineq4}
 C\left(\rho_{A B}\right)=\min _{\left\{p_i,\left|\psi_i\right\rangle\right\}} \sum_i p_i C\left(\left|\psi_i\right\rangle\right),\quad
C_{a}\left(\rho_{A B}\right)=\max _{\left\{p_i,\left|\psi_i\right\rangle\right\}} \sum_i p_i C\left(\left|\psi_i\right\rangle\right),
\end{equation}
where the minimum/maximum are taken over all possible pure decompositions of $\rho_{A B}=\sum_i p_i\left|\psi_i\right\rangle\left\langle\psi_i\right|$ with $p_i \geqslant 0, \sum_i p_i=1$ and $\left|\psi_i\right\rangle \in \mathcal{H}_A \otimes \mathcal{H}_B$.

\bigskip
The following result is directly derived from Theorem \ref{t:chap2-thm1}.

\begin{cor} \label{c:chap2-cor1} Let $C$ be a bipartite entanglement measure concurrence satisfying the $2$nd-mono\-gamy relation \eqref{e:chap1-ineq1} and $\rho_{AB_{1}\cdots B_{N-1}}$ any $N$-qubit quantum state. Arrange \{$C_{i}=C_{AB_{i^{\prime}}}|i=1,\cdots,N-1\}$ in descending order
such that $C^{2}_{i}\geq C^{2}_{i+1}>0$ for $i=1,\cdots,N-2$, then for $0\leq\alpha\leq 2$ we have
\begin{equation}\label{e:chap2-ineq5}
\begin{aligned}
C^{\alpha}_{A|B_{1}\cdots B_{N-1}} \geq \sum_{k=1}^{N-2}\left(1+\frac{k-1}{s}\right)^{\frac{\alpha}{2}-1}\prod_{j=2}^{N-k}\left(1+\frac{s}{N-j}\right)^{\frac{\alpha}{2}-1}C_{k}^{\alpha}+\left(1+\frac{N-2}{s}\right)^{\frac{\alpha}{2}-1}C_{N-1 }^{\alpha}.
\end{aligned}
\end{equation}
where $q\leq s \leq 1$ and $q=max\left\{\frac{hC_{h+1}^{2}}{C_{1}^{2}+\cdots+C_{h}^{2}}|~~h=1,2,\cdots,N-2\right\}.$
\end{cor}

\bigskip
\begin{exmp} Let $\rho=|\Phi\rangle_{AB_{1}B_{2}B_{3}}\langle\Phi|$ be a $4$-qubit entangled decoherence-free state \cite{YL}:
$$|\Phi\rangle_{AB_{1}B_{2}B_{3}}=\frac{\sqrt{2}}{2}|\Phi_{0}\rangle_{AB_{1}B_{2}B_{3}}+\frac{\sqrt{2}}{2}|\Phi_{1}\rangle_{AB_{1}B_{2}B_{3}},$$
where $|\Phi_{0}\rangle_{AB_{1}B_{2}B_{3}}=\frac{1}{2}(|01\rangle-|10\rangle)_{AB_{1}}(|01\rangle-|10\rangle)_{B_{2}B_{3}}$, $|\Phi_{1}\rangle_{AB_{1}B_{2}B_{3}}=\frac{1}{2\sqrt{3}}(2|1100\rangle+2|0011\rangle-|1010\rangle-|1001\rangle-|0101\rangle-|0110\rangle)_{AB_{1}B_{2}B_{3}}$.
Then
$C_{AB_{1}}= 0.9107, C_{AB_{2}}=0.3333, C_{AB_{3}}=0.244$.  Set $s = 0.6$ (since $q\leq s \leq 1$ and $q=max\left\{\frac{C_{AB_{2}}^{2}}{C_{AB_{1}}^{2}}, \frac{2C_{AB_{3}}^{2}}{C_{AB_{1}}^{2}+C_{AB_{2}}^{2}}\right\}=0.5359$).

For $0\leq\alpha\leq 2$, Corollary \ref{c:chap2-cor1}
implies that the RHS of our monogamy relation is:
\begin{align*}
X_1&=\left(1+\frac{2}{s}\right)^{\frac{\alpha}{2}-1}C_{AB_{3}}^{\alpha}+\left(1+\frac{s}{2}\right)^{\frac{\alpha}{2}-1}\left(1+\frac{1}{s}\right)^{\frac{\alpha}{2}-1}C_{AB_{2}}^{\alpha}+\left(1+\frac{s}{2}\right)^{\frac{\alpha}{2}-1}\left(1+\frac{s}{1}\right)^{\frac{\alpha}{2}-1}C_{AB_{1}}^{\alpha}\\
&=4.3333^{\frac{\alpha}{2}-1}0.244^{\alpha}+3.4667^{\frac{\alpha}{2}-1}0.3333^{\alpha}+2.08^{\frac{\alpha}{2}-1}0.9107^{\alpha}.
\end{align*}

The RHS $X_2$ of the monogamy relation derived from \cite[Lem. 1]{ZLJM} is a special case of our bound at $s=1$:
\begin{align*}
 X_2=3^{\frac{\alpha}{2}-1}\left(C_{AB_{3}}^{\alpha}+C_{AB_{2}}^{\alpha}+C_{AB_{1}}^{\alpha}\right)=3^{\frac{\alpha}{2}-1}\left(0.244^{\alpha}+0.3333^{\alpha}+0.9107^{\alpha}\right).
\end{align*}

The RHS $X_3$ of the monogamy relation from \cite[Lem. 1]{JFQ} is:
\begin{align*}
X_3&=C_{AB_{3}}^{\alpha}+\frac{3^{\frac{\alpha}{2}}-1}{2^{\frac{\alpha}{2}}}C_{AB_{2}}^{\alpha}+\frac{3^{\frac{\alpha}{2}}-1}{2^{\frac{\alpha}{2}}}\left(2^{\frac{\alpha}{2}}-1\right)C_{AB_{1}}^{\alpha}\\
&=0.244^{\alpha}+\frac{3^{\frac{\alpha}{2}}-1}{2^{\frac{\alpha}{2}}}0.3333^{\alpha}+\frac{3^{\frac{\alpha}{2}}-1}{2^{\frac{\alpha}{2}}}\left(2^{\frac{\alpha}{2}}-1\right)0.9107^{\alpha}.
\end{align*}

The lower bound $X_{_4}$ of the monogamy relation obtained from  \cite[Lem. 1]{ZF3} is:
\begin{align*}
 X_{_4}&=C_{AB_{3}}^{\alpha}+\left(2^{\frac{\alpha}{2}}-1\right)C_{AB_{2}}^{\alpha}+\left(2^{\frac{\alpha}{2}}-1\right)^{2}C_{AB_{1}}^{\alpha}\\
 &=0.244^{\alpha}+\left(2^{\frac{\alpha}{2}}-1\right)0.3333^{\alpha}+\left(2^{\frac{\alpha}{2}}-1\right)^{2}0.9107^{\alpha}.
\end{align*}
\begin{figure}[H]
\caption{Comparison of Monogamy Bounds I}
    \centering
    \includegraphics[width=0.75\textwidth]{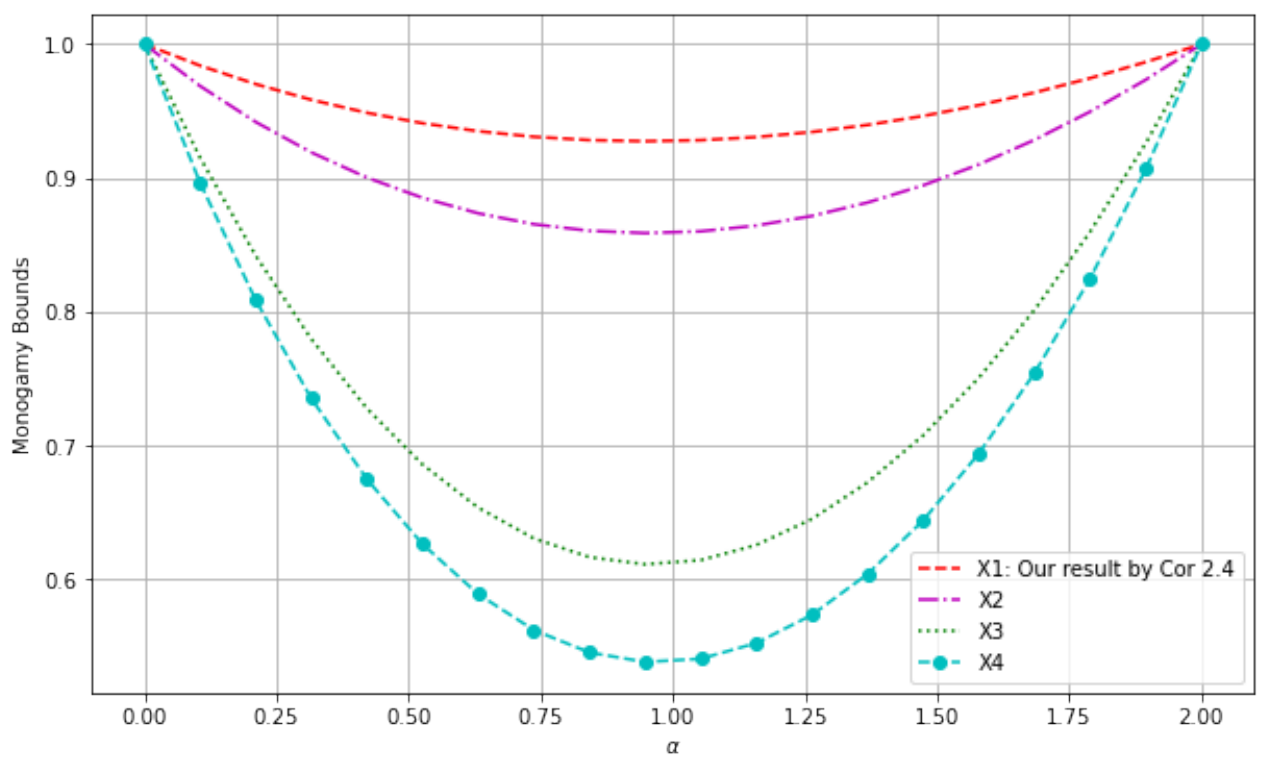}
    \captionsetup{justification=raggedright} 
    \caption*{The graphs of $X_1, X_2, X_3$ and $X_4$ (cf. the legend)
are shown in Figure 1 from top to bottom, which shows
that our bound $X_1$ from Cor. \ref{c:chap2-cor1} is the highest compared with those
from \cite[Lem. 1]{ZLJM}, \cite[Lem. 1]{JFQ} and \cite[Lem. 1]{ZF3} respectively.} 
\end{figure}
\begin{figure}[H]
\caption{Comparison of Monogamy Bounds II}
    \centering
    \includegraphics[width=1.0\textwidth]{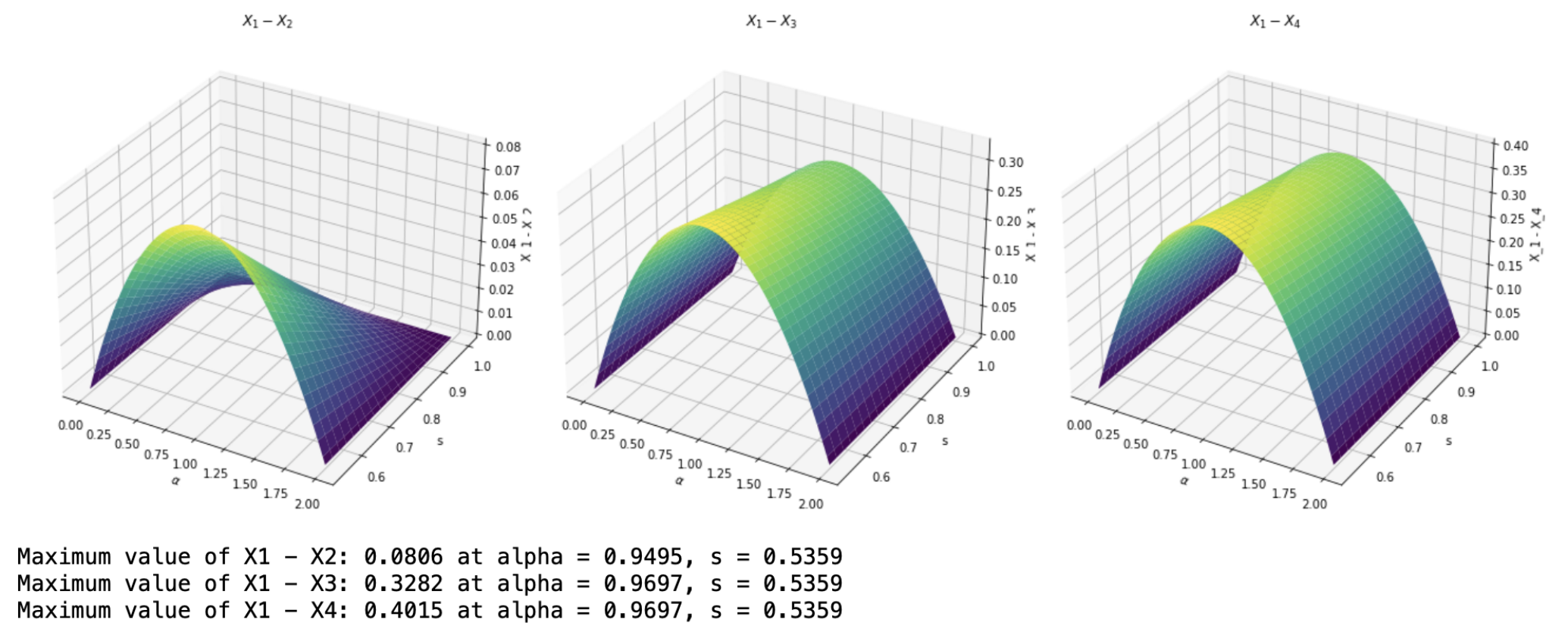}
    \captionsetup{justification=raggedright} 
    \caption*{We have also drawn the differences in Fig. 2, which further confirms that our lower bound $X_1$ is the best. The maxima of the differences are indicated.} 
\end{figure}

\end{exmp}

\section{\textbf{Tighter $\beta$th $(\beta \geq \delta)$ power polygamy relations of assisted entanglement }}\label{s:polygamy1}
 In this section, we will present a new class of $\beta$th $(\beta \geq \delta)$ power polygamy relations for any $N$-qubit quantum state in a unified manner. First of all, we need to recall the following lemma from \cite{CJW}.

\begin{lem}\cite{CJW}\label{l:chap3-lem1}
Let $a\geq 1$ be a real number. Then for $t\geq a\geq 1$ and $ x\geq 1$, we have
\begin{equation}\label{e:chap3-ineq1}
\begin{aligned}
(1+t)^{x}&\leq\left(1+\frac{a}{s}\right)^{x-1}+\left(1+\frac{s}{a}\right)^{x-1}t^{x} \leq\left(1+a\right)^{x-1}+\left(1+\frac{1}{a}\right)^{x-1}t^{x}\\
&\leq 1+\frac{(1+a)^{x}-1}{a^{x}}t^{x} \leq 1+(2^{x}-1)t^{x}
\end{aligned}
\end{equation}
for any real $s$ satisfying $\frac{a}{t}\leq s \leq 1$. 
\end{lem}

Next, we give an analogue of Lemma \ref{l:chap2-lem2}.

\begin{lem}\label{l:chap3-added lem} Let $p_i$ be $N$ positive numbers such that $p_{i}\geq  p_{i+1}(i=1,\cdots,N-1)$, then one has that
\begin{equation}
\begin{aligned}
\left(\sum_{i=1}^{N}p_{i}\right)^{x} \leq \sum_{k=1}^{N-1}\left(1+\frac{k-1}{s}\right)^{x-1}\prod_{j=1}^{N-k}\left(1+\frac{s}{N-j}\right)^{x-1}p_{k}^{x}+\left(1+\frac{N-1}{s}\right)^{x-1}p_{N}^{x}
\end{aligned}
\end{equation}
for $x\geq 1$, where $r\leq s \leq 1$ and $r=max\left\{\frac{hp_{h+1}}{p_{1}+\cdots+p_{h}}|~~h=1,\cdots,N-1\right\}.$ 
\end{lem}

Similar to Theorem \ref{t:chap2-thm1}, we have the following conclusion by using Lemma \ref{l:chap3-added lem}.


\begin{thm} \label{t:chap3-thm1}
Let $\mathcal{E}_{a}$ be a bipartite assisted entanglement measure satisfying the $\delta$th-polygamy relation \eqref{e:chap1-ineq2} and $\rho_{AB_{1}\cdots B_{N-1}}$ any $N$-qubit quantum state. Arrange $\{\mathcal{E}_{a_{i}}=\mathcal{E}_{aAB_{i'}}|i=1,\cdots,N-1\}$ in descending order. If $\mathcal{E}^{\delta}_{a_{i}}\geq \mathcal{E}^{\delta}_{a_{i+1}}>0$ for $i=1,\cdots,N-2$, then
\begin{align*}
\mathcal{E}^{\beta}_{aA|B_{1}\cdots B_{N-1}} \leq \sum_{k=1}^{N-2}\left(1+\frac{k-1}{s}\right)^{\frac{\beta}{\delta}-1}\prod_{j=2}^{N-k}\left(1+\frac{s}{N-j}\right)^{\frac{\beta}{\delta}-1}\mathcal{E}_{a_{k}}^{\beta}+\left(1+\frac{N-2}{s}\right)^{\frac{\beta}{\delta}-1}\mathcal{E}_{a_{N-1} }^{\beta}
\end{align*}
for $\beta\geq \delta$, where $\tilde{q}\leq s \leq 1$ and $\tilde{q}=max\left\{\frac{h\mathcal{E}^{\delta}_{a_{h+1}}}{\mathcal{E}^{\delta}_{a_{1}}+\cdots+\mathcal{E}^{\delta}_{a_{h}}}|~~h=1,2,\cdots,N-2\right\}.$
\end{thm}

{\bf Comparison of the polygamy relations for assisted entanglement measure $\mathcal{E}_{a}$}. Based on Theorem \ref{t:chap3-thm1} and Lemma \ref{l:chap3-lem1}, we have the following strong unified polygamy relations of $\beta$th $( \beta\geq \delta)$ power of $\mathcal{E}_{a}$.

\begin{align*}
\mathcal{E}^{\beta}_{aA|B_{1}\cdots B_{N-1}} & \leq \sum_{k=1}^{N-2}\left(1+\frac{k-1}{s}\right)^{\frac{\beta}{\delta}-1}\prod_{j=2}^{N-k}\left(1+\frac{s}{N-j}\right)^{\frac{\beta}{\delta}-1}\mathcal{E}_{a_{k}}^{\alpha}+\left(1+\frac{N-2}{s}\right)^{\frac{\alpha}{\delta}-1}\mathcal{E}_{a_{N-1} }^{\beta}\\
&\leq \sum_{k=1}^{N-2}k^{\frac{\beta}{\delta}-1}\prod_{j=2}^{N-k}\left(1+\frac{1}{N-j}\right)^{\frac{\beta}{\delta}-1}\mathcal{E}_{a_{k}}^{\alpha}+\left(N-1\right)^{\frac{\alpha}{\delta}-1}\mathcal{E}_{a_{N-1} }^{\beta}\\
&\leq \sum_{k=1}^{N-2}\prod_{j=2}^{N-k}\frac{(N-j+1)^{\frac{\beta}{\delta}}-1}{(N-j)^{\frac{\beta}{\delta}}}\mathcal{E}_{a_{k}}^{\beta}+\mathcal{E}_{a_{N-1} }^{\beta}\\
&\leq \sum_{k=1}^{N-1}(2^{\frac{\beta}{\delta}}-1)^{N-1-k}\mathcal{E}_{a_{k}}^{\beta}
\end{align*}
where $s, \tilde{q}$ are defined as in Theorem \ref{t:chap3-thm1}.

\bigskip
From Theorem \ref{t:chap3-thm1}, we can derive the following corollary.

\begin{cor} \label{c:chap3-cor1} Let $|\psi\rangle_{AB_{1}\cdots B_{N-1}}$ be any $N$-qubit pure state and $C_{a}$ the bipartite assisted quantum measure CoA satisfying the $2$nd-polygamy relation \eqref{e:chap1-ineq2}. Rename $C_{a_{i}}=C_{aAB_{i^{\prime}}}$ so that $C^{2}_{a_{i}}\geq C^{2}_{a_{i+1}}>0$ for $i=1,\cdots,N-2$, then for $\beta\geq 2$ we have
\begin{align*}
C^{\beta}_{a}(|\psi\rangle_{A|B_{1}\cdots B_{N-1}}) \leq \sum_{k=1}^{N-2}\left(1+\frac{k-1}{s}\right)^{\frac{\beta}{2}-1}\prod_{j=2}^{N-k}\left(1+\frac{s}{N-j}\right)^{\frac{\beta}{2}-1}C_{a_{k}}^{\beta}+\left(1+\frac{N-2}{s}\right)^{\frac{\beta}{2}-1}C_{a_{N-1} }^{\beta}.
\end{align*}
where $\tilde{q}\leq s \leq 1$ and $\tilde{q}=max\left\{\frac{hC_{a_{h+1}}^{2}}{C_{a_{1}}^{2}+\cdots+C_{a_{h}}^{2}}|~~h=1,2,\cdots,N-2\right\}.$
\end{cor}

\bigskip
\begin{exmp}
Consider the following $4$-qubit generalized $W$-class state \cite{YCLZ}:
$$
|W\rangle_{AB_{1}B_{2}B_{3}}=\lambda_{1}(|1000\rangle+\lambda_{2}|0100\rangle)+\lambda_{3}|0010\rangle+\lambda_{4}|0001\rangle.
$$
 where $\sum_{i=1}^4 \lambda_i^2=1$, and $\lambda_{i}\geq 0$ for $i=1,2,3,4$.
Then \cite{YCLZ} implies that $C_{aAB_{1}}=2 \lambda_1\lambda_2 $, $C_{aAB_{2}}=2 \lambda_1\lambda_3 $, $C_{aAB_{3}}=2 \lambda_1\lambda_4 $. Set $\lambda_1=\frac{3}{4}, \lambda_2=\frac{1}{2}, \lambda_3=\frac{\sqrt{2}}{4}, \lambda_4=\frac{1}{4}$, we have
$C_{aAB_{1}}=\frac{3}{4}, C_{aAB_{2}}=\frac{3\sqrt{2}}{8}, C_{aAB_{3}}=\frac{3}{8}.$ Set $s = \frac{3}{5}$ (since $\tilde{q}\leq s \leq 1$ and $\tilde{q}=\frac{1}{2}$).

Therefore, by Corollary \ref{c:chap3-cor1}, for $\beta\geq 2$, our upper bound of the polygamy relation is
\begin{align*}
Y_1&=\left(1+\frac{2}{s}\right)^{\frac{\beta}{2}-1}C_{aAB_{3}}^{\beta}+\left(1+\frac{s}{2}\right)^{\frac{\beta}{2}-1}\left(1+\frac{1}{s}\right)^{\frac{\beta}{2}-1}C_{aAB_{2}}^{\beta}+\left(1+\frac{s}{2}\right)^{\frac{\beta}{2}-1}\left(1+\frac{s}{1}\right)^{\frac{\beta}{2}-1}C_{aAB_{1}}^{\beta}\\
&=\left(\frac{13}{3}\right)^{\frac{\beta}{2}-1}\left(\frac{3}{8}\right)^{\beta}+\left(\frac{52}{15}\right)^{\frac{\beta}{2}-1}\left(\frac{3\sqrt{2}}{8}\right)^{\beta}+\left(\frac{52}{25}\right)^{\frac{\beta}{2}-1}\left(\frac{3}{4}\right)^{\beta}.
\end{align*}

The upper bound $Y_2$ of the polygamy relation in
\cite[Lem. 3]{ZLJM} is a special case of our bound at $s=1$:
\begin{equation*}
     Y_2=3^{\frac{\beta}{2}-1}C_{aAB_{3}}^{\beta}+\left(\frac{3}{2}\right)^{\frac{\beta}{2}-1}2^{\frac{\beta}{2}-1}C_{aAB_{2}}^{\beta}+\left(\frac{3}{2}\right)^{\frac{\beta}{2}-1}2^{\frac{\beta}{2}-1}C_{aAB_{1}}^{\beta}=3^{\frac{\beta}{2}-1}\left(\left(\frac{3}{8}\right)^{\beta}+\left(\frac{3\sqrt{2}}{8}\right)^{\beta}+\left(\frac{3}{4}\right)^{\beta}\right).
\end{equation*}

The upper bound $Y_3$ of the polygamy relation from \cite[Lem. 2]{JFQ} is:
\begin{equation*}
    Y_3=C_{aAB_{3}}^{\beta}+\frac{3^{\frac{\beta}{2}}-1}{2^{\frac{\beta}{2}}}C_{aAB_{2}}^{\beta}+\frac{3^{\frac{\beta}{2}}-1}{2^{\frac{\beta}{2}}}\left(2^{\frac{\beta}{2}}-1\right)C_{aAB_{1}}^{\beta}=\left(\frac{3}{8}\right)^{\beta}+\frac{3^{\frac{\beta}{2}}-1}{2^{\frac{\beta}{2}}}\left(\frac{3\sqrt{2}}{8}\right)^{\beta}+\frac{3^{\frac{\beta}{2}}-1}{2^{\frac{\beta}{2}}}\left(2^{\frac{\beta}{2}}-1\right)\left(\frac{3}{4}\right)^{\beta}.
\end{equation*}

The upper bound $Y_{_4}$ of the polygamy relation which can be obtained using $(1+t)^{x}\leq 1+(2^{x}-1)$ ($t,x\geq 1$) from \cite{ZF3} is:
\begin{equation*}
    Y_{_4}=C_{aAB_{3}}^{\beta}+\left(2^{\frac{\beta}{2}}-1\right)C_{aAB_{2}}^{\beta}+\left(2^{\frac{\beta}{2}}-1\right)^{2}C_{aAB_{1}}^{\beta}=\left(\frac{3}{8}\right)^{\beta}+\left(2^{\frac{\beta}{2}}-1\right)\left(\frac{3\sqrt{2}}{8}\right)^{\beta}+\left(2^{\frac{\beta}{2}}-1\right)^{2}\left(\frac{3}{4}\right)^{\beta}.
\end{equation*}
\begin{figure}[H]
\caption{Comparison of Polygamy Relations I}
    \centering
    \includegraphics[width=0.75\textwidth]{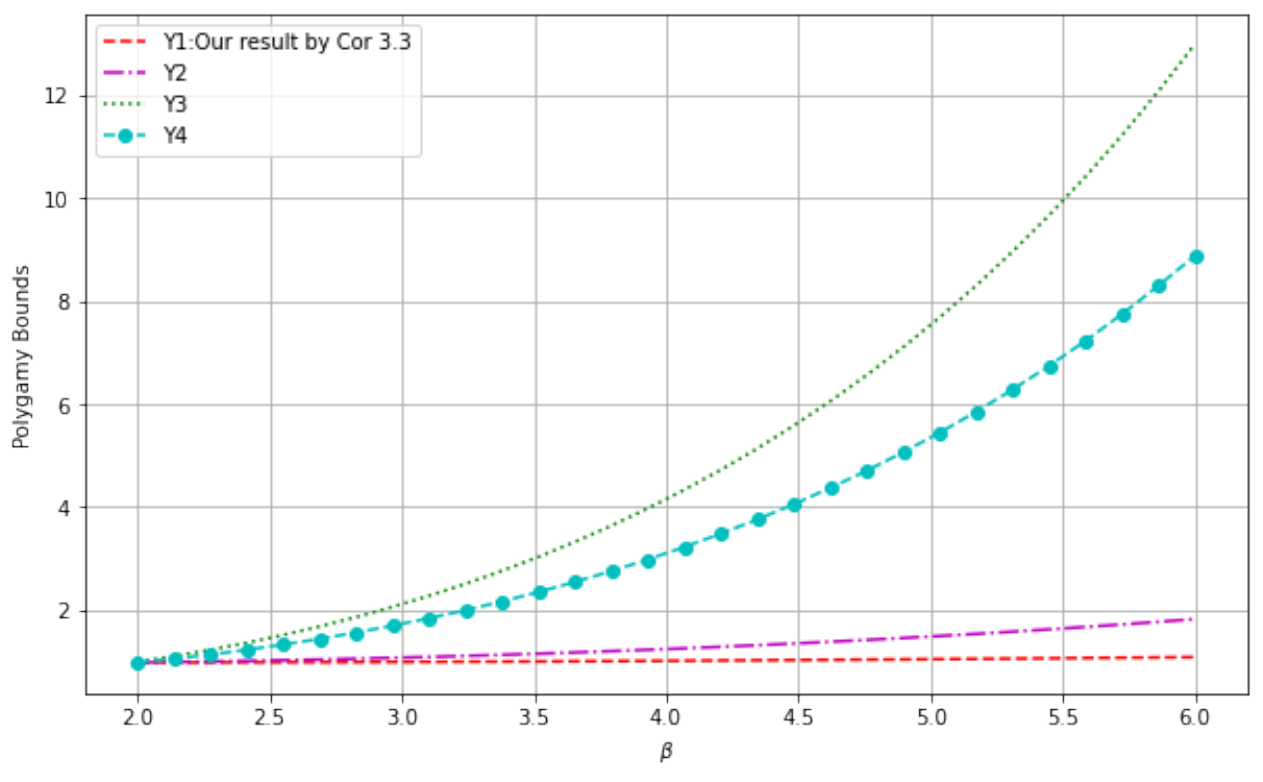}
    \captionsetup{justification=raggedright} 
    \caption*{The graphs of $Y_1, Y_2, Y_3$ and $Y_4$ (cf. the legend)
are shown in Figure 3 from bottom to top, which shows
that our bound $Y_1$ from Cor. \ref{c:chap3-cor1} is the lowest compared with those
from \cite[Lem. 3]{ZLJM}, \cite[Lem. 2]{JFQ} and \cite{ZF3} respectively.
} 
\end{figure}
\begin{figure}[H]
\caption{Comparison of polygamy relations II}
    \centering
    \includegraphics[width=1.0\textwidth]{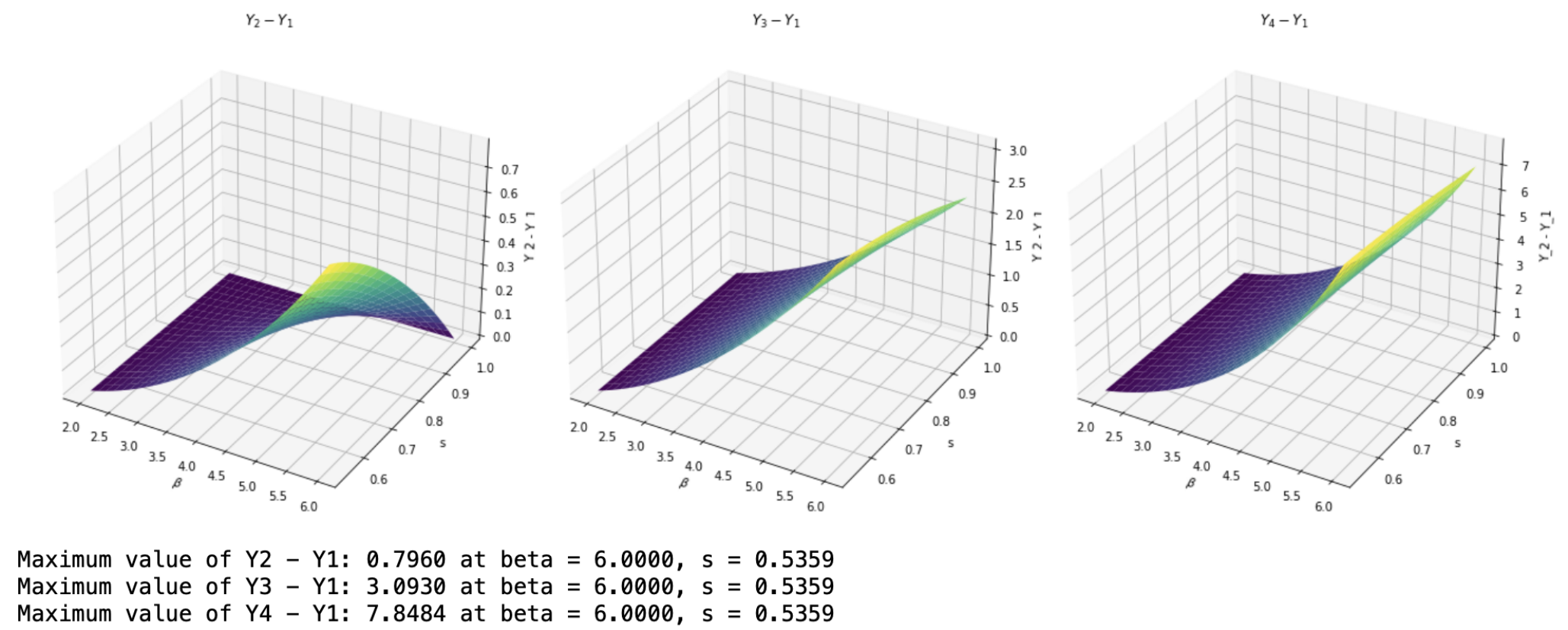}
    \captionsetup{justification=raggedright} 
    \caption*{We also added pictures of their differences in Fig. 4, which shows that the upper bound $Y_1$ performs best. The maxima of the differences are also marked.
} 
\end{figure}
\end{exmp}

\bigskip
\section{\textbf{Tighter $\omega$th $(0\leq \omega \leq \delta)$ power polygamy relations of assisted entanglement }}\label{s:polygamy2}

In this section, we present a class of tight polygamy inequalities of $\omega$th $(0\leq \omega \leq \delta)$ power of assisted entanglement measures $\mathcal{E}_{a}$ based on Eq. \eqref{e:chap1-ineq2} in a unified manner.
Then we use this information to derive
bounds for the $\omega$th ($0\leq \omega \leq 2$) power of concurrence for any $N$-qubit pure states $|\psi\rangle_{AB_{1}\cdots B_{N-1}}$ under the partition $AB_{1}$ and $B_{2}\cdots B_{N-1}$, which would give some
estimate of the linear entropy.

We remark that our treatment works for an arbitrary measure.

\subsection{\textbf{ The $\omega$th $(0\leq \omega \leq \delta)$ power polygamy relations}}

First, we need the following lemmas.

\begin{lem}\label{l:chap4-lem1}
Let $t\geq m\geq1$ and $0\leq x \leq 1$, then
\begin{equation}\label{e:chap4-ineq1}
\begin{aligned}
(1+t)^{x}\leq t^{x}+(1+m)^{x}-m^{x}+\frac{xm^{2}}{(1+m)^{2}}\left(t^{x-1}-m^{x-1}\right).
\end{aligned}
\end{equation}
\end{lem}

\begin{proof}
Fix $m (\geq1)$, let
$f(x,y)=(1+y)^{x}-y^{x}-\frac{m^{2}}{(1+m)^{2}}xy^{x-1}$ defined on $(x,y)\in [0,1]\times[m,+\infty)$. Then $\frac{\partial f(x,y)}{\partial y}=xy^{x-1}\left(\left(1+\frac{1}{y}\right)^{x-1}-(x-1)\frac{m^{2}}{(1+m)^{2}}\frac{1}{y}-1\right).$ Let $h(x,y)=\left(1+\frac{1}{y}\right)^{x-1}-(x-1)\frac{m^{2}}{(1+m)^{2}}\frac{1}{y}-1$, $( y \geq m, 0\leq x\leq 1)$, then we have $\frac{\partial h(x,y)}{\partial y}=\frac{-1}{y^{2}}(x-1)\left[(1+\frac{1}{y})^{x-2}-\frac{m^{2}}{(1+m)^{2}}\right]\geq 0$. This means that for $y \geq m, 0\leq x\leq 1$, the function $h(x,y)$ is increasing  with respect to $y$. Subsequently, we have $h(x,y)\leq h(x,+\infty)=\lim_{y\rightarrow +\infty}h(x,y)=0.$ Therefore $\frac{\partial f(x,y)}{\partial y}\leq 0$, and
$f(x,y)$ is decreasing as a function of $y$. Thus $f(x,t)\leq f(x,m)$ for $t\geq m$, which is \eqref{e:chap4-ineq1}.
\end{proof}

Note that $t^{x-1}-m^{x-1}\leq 0$ for $t\geq m\geq 1$ and $0 \leq x\leq 1$, thus
\begin{equation}\label{e:chap4-ineq2}
\begin{aligned}
(1+t)^{x} &\leq t^{x}+(1+m)^{x}-m^{x}+\frac{xm^{2}}{(1+m)^{2}}\left(t^{x-1}-m^{x-1}\right)\\
&\leq t^{x}+(1+m)^{x}-m^{x}\leq t^{x}+2^{x}-1 \leq t^{x}+x \leq t^{x}+1.
\end{aligned}
\end{equation}

\begin{lem}\label{l:chap4-lem2} Let $p_i$ be $N$ positive numbers such that $p_{i}\geq  p_{i+1}(i=1,\cdots,N-1)$, then
\begin{equation}\label{e:chap4-ineq3}
\begin{aligned}
\left(\sum_{i=1}^{N}p_{i}\right)^{x}
&\leq p_{1}^{x}+(2^{x}-1)p_{2}^{x}+(3^{x}-2^{x})p_{3}^{x}+\cdots+\left[N^{x}-(N-1)^{x}\right]p_{N}^{x}\\
&+\frac{1}{2^{2}}x\left(\left(\frac{p_{1}}{p_{2}}\right)^{x-1}-1\right)+\frac{2^{2}}{3^{2}}xp_{3}^{x}\left(\left(\frac{p_{1}+p_{2}}{p_{3}}\right)^{x-1}-2^{x-1}\right)p_{2}^{x}\\
&+\frac{(N-1)^{2}}{N^{2}}x\left(\left(\frac{p_{1}+p_{2}+\cdots+p_{N-1}}{p_{N}}\right)^{x-1}-\left(N-1\right)^{x-1}\right)p_{N}^{x}.
\end{aligned}
\end{equation}
\end{lem}
\begin{proof} We use induction on $N$. The case of $N = 1$ is trivial. Assume Eq. \eqref{e:chap4-ineq3} holds for $<N$. Now
consider $N$ decreasing positive numbers $p_1\geq p_2\geq \ldots\geq p_N > 0$.
Setting $t=\frac{p_{1}+p_{2}+\cdots+p_{N-1}}{p_{N}}\geq N-1$, Lemma \ref{l:chap4-lem1} implies that
\begin{align*}
\left(\sum_{i=1}^{N}p_{i}\right)^{x}  &=\left(p_{1}+p_{2}+\cdots+p_{N}\right)^{x}=p_{N}^{x}\left(1+\frac{p_{1}+p_{2}+\cdots+p_{N-1}}{p_{N}}\right)^{x}\\
&\leq (p_{1}+p_{2}+\cdots+p_{N-1})^{x}+\left[N^{x}-(N-1)^{x}\right)]p_{N}^{x}\\
&+\frac{(N-1)^{2}}{N^{2}}x\left(\left(\frac{p_{1}+p_{2}+\cdots+p_{N-1}}{p_{N}}\right)^{x-1}-\left(N-1\right)^{x-1}\right)p_{N}^{x}.
\end{align*}
Thus, the inequality \eqref{e:chap4-ineq3} follows by induction.
\end{proof}

According to Lemma \ref{l:chap4-lem2},  we get the following polygamy relations for any $N$-qubit quantum state $\rho_{AB_{1}\cdots B_{N-1}}$.

\begin{thm} \label{t:chap4-thm1}
Let $\mathcal{E}_{a}$ be a bipartite assisted quantum measure satisfying the $\delta$-polygamy relation \eqref{e:chap1-ineq2} and $\rho_{AB_{1}\cdots B_{N-1}}$  any $N$-qubit quantum state. Arrange \{$\mathcal{E}_{a_{i}}=\mathcal{E}_{aAB_{i'}}|i=1,\cdots,N-1\}$ in descending order so that
$\mathcal{E}^{\delta}_{a_{1}}\geq \mathcal{E}^{\delta}_{a_{2}}\geq \ldots \geq \mathcal{E}^{\delta}_{a_{N-1}}>0$, then
\begin{align*}
\mathcal{E}_{aA|B_{1}\cdots B_{N-1}}^{\omega}\leq \sum_{k=1}^{N-1}\left[k^{\frac{\omega}{\delta}}-(k-1)^{\frac{\omega}{\delta}}\right]\mathcal{E}_{a_{k}}^{\omega}+\frac{\omega}{\delta}\sum_{p=2}^{N-1}\frac{(p-1)^{2}}{p^{2}}\left[\tau_{p}^{\frac{\omega}{\delta}-1}-(p-1)^{\frac{\omega}{\delta}-1}\right]\mathcal{E}_{a_{p}}^{\omega}
\end{align*}
for $0\leq \omega \leq \delta$, where $\tau_{p}=\frac{\mathcal{E}_{a_{1}}^{\delta}+\cdots+\mathcal{E}_{a_{p-1}}^{\delta}}{\mathcal{E}_{a_{p}}^{\delta}}$, $p=2,\cdots,N-1$.
\end{thm}
\begin{proof} From the $\delta$th-polygamy relation \eqref{e:chap1-ineq2} and Lemma \ref{l:chap4-lem2} we have
\begin{align*}
&\mathcal{E}_{aA|B_{1}\cdots B_{N-1}}^{\omega}\leq (\mathcal{E}_{aAB_{1}}^{\delta}+\mathcal{E}_{aAB_{2}}^{\delta}+\cdots+\mathcal{E}_{aAB_{N-1}}^{\delta})^{\frac{\omega}{\delta}}= (\mathcal{E}_{a_{1}}^{\delta}+\mathcal{E}_{a_{2}}^{\delta}+\cdots+\mathcal{E}_{a_{N-1}}^{\delta})^{\frac{\omega}{\delta}}\\
&\leq \mathcal{E}_{a_{1}}^{\omega}+(2^{\frac{\omega}{\delta}}-1)\mathcal{E}_{a_{2}}^{\omega}+\cdots+\left[(N-1)^{\frac{\omega}{\delta}}-(N-2)^{\frac{\omega}{\delta}}\right]\mathcal{E}_{a_{N-1}}^{\omega}
+\frac{1}{2^{2}}\frac{\omega}{\delta}\left[\left(\frac{\mathcal{E}_{a_{1}}^{\delta}}{\mathcal{E}_{a_{2}}^{\delta}}\right)^{\frac{\omega}{\delta}-1}-1\right]\mathcal{E}_{a_{2}}^{\omega}\\
&+\cdots+\frac{(N-2)^{2}}{(N-1)^{2}}\frac{\omega}{\delta}\left[\left(\frac{\mathcal{E}_{a_{1}}^{\delta}+\cdots+\mathcal{E}_{a_{N-2}}^{\delta}}{\mathcal{E}_{a_{N-1}}^{\delta}}\right)^{\frac{\omega}{\delta}-1}-(N-2)^{\frac{\omega}{\delta}-1}\right]\mathcal{E}_{a_{N-1}}^{\omega}\\
&=\sum_{k=1}^{N-1}\left[k^{\frac{\omega}{\delta}}-(k-1)^{\frac{\omega}{\delta}}\right]\mathcal{E}_{a_{k}}^{\omega}+\frac{\omega}{\delta}\sum_{p=2}^{N-1}\frac{(p-1)^{2}}{p^{2}}\left[\tau_{p}^{\frac{\omega}{\delta}-1}-(p-1)^{\frac{\omega}{\delta}-1}\right]\mathcal{E}_{a_{p}}^{\omega}.
\end{align*}
\end{proof}

{\bf Comparison of the polygamy relations for assisted entanglement measure $\mathcal{E}_{a}$}.
Based on Theorem \ref{t:chap4-thm1} and Eq. \eqref{e:chap4-ineq2}, we obtain the following unified polygamy relations of $\omega$th $(0\leq \omega\leq \delta)$ power of $\mathcal{E}_{a}$.

\begin{align*}
\mathcal{E}_{aA|B_{1}\cdots B_{N-1}}^{\omega}&\leq \sum_{k=1}^{N-1}\left[k^{\frac{\omega}{\delta}}-(k-1)^{\frac{\omega}{\delta}}\right]\mathcal{E}_{a_{k}}^{\omega}+\frac{\omega}{\delta}\sum_{p=2}^{N-1}\frac{(p-1)^{2}}{p^{2}}\left[\tau_{p}^{\frac{\omega}{\delta}-1}-(p-1)^{\frac{\omega}{\delta}-1}\right]\mathcal{E}_{a_{p}}^{\omega}\\
&\leq \sum_{k=1}^{N-1}\left[k^{\frac{\omega}{\delta}}-(k-1)^{\frac{\omega}{\delta}}\right]\mathcal{E}_{a_{k}}^{\omega}\leq \mathcal{E}_{a_{1}}^{\omega}+\sum_{k=2}^{N-1}\left[2^{\frac{\omega}{\delta}}-1\right]\mathcal{E}_{a_{k}}^{\omega}\\
&\leq \mathcal{E}_{a_{1}}^{\omega}+ \sum_{k=2}^{N-1}\frac{\omega}{\delta}\mathcal{E}_{a_{k}}^{\omega}\leq \sum_{k=1}^{N-1}\mathcal{E}_{a_{k}}^{\omega}
\end{align*}
 where $\tau_{p}$ $(p=2,\cdots,N-1)$ are defined as in Theorem \ref{t:chap4-thm1}.

\bigskip

In view of the comparison, we have the following polygamy relations of $\omega$th $(0\leq \omega \leq 2)$ power of CoA.

\begin{cor} \label{c:chap4-cor1}
Let $|\psi\rangle_{AB_{1}\cdots B_{N-1}}$ be any $N$-qubit pure state and $C_{a}$ be bipartite assisted quantum measure CoA satisfying the $2$-polygamy relation \eqref{e:chap1-ineq2}. Arrange \{$C_{a_{i}}=C_{aAB_{i^{\prime}}}|i=1,\cdots,N-1\}$ in descending order such that $C^{2}_{a_{i}}\geq C^{2}_{a_{i+1}}>0$ for $i=1,\cdots,N-2$, then
\begin{align}
C^{\omega}(|\psi\rangle_{A|B_{1}\cdots B_{N-1}})&\leq \sum_{k=1}^{N-1}\left[k^{\frac{\omega}{2}}-(k-1)^{\frac{\omega}{2}}\right]C_{a_{k}}^{\omega}+\frac{\omega}{2}\sum_{v=2}^{N-1}\frac{(v-1)^{2}}{v^{2}}\left[\tau_{v}^{\frac{\omega}{2}-1}-(v-1)^{\frac{\omega}{2}-1}\right]C_{a_{v}}^{\omega} \label{e:chap4-ineq4}\\
&\leq C_{a_{1}}^{\omega}+\sum_{k=2}^{N-1}\left[k^{\frac{\omega}{2}}-(k-1)^{\frac{\omega}{2}}\right]C_{a_{k}}^{\omega} \quad(\text{by $\tau_{v}^{\frac{\omega}{2}-1}-(v-1)^{\frac{\omega}{2}-1}\leq 0$}) \label{e:chap4-ineq5}\\
&\leq C_{a_{1}}^{\omega}+\sum_{k=2}^{N-1}\left[2^{\frac{\omega}{2}}-1\right]C_{a_{k}}^{\omega} \quad(\text{by \cite[Eq. (9)]{JF2}}) \label{e:chap4-ineq6}\\
&\leq C_{a_{1}}^{\omega}+\sum_{k=2}^{N-1}\frac{\omega}{2}C_{a_{k}}^{\omega} \quad(\text{by \cite[Eq. (8)]{JFL}}) \label{e:chap4-ineq7}\\
&\leq \sum_{k=1}^{N-1}C_{a_{k}}^{\omega} \quad\quad\quad(\text{by \cite[Conjecture 2] {GBS}})   \label{e:chap4-ineq8}
\end{align}
for $0\leq \omega \leq 2$, where the first inequality \eqref{e:chap4-ineq4} follows by Theorem \ref{t:chap4-thm1}, and $\tau_{v}=\frac{C_{a_{1}}^{2}+\cdots+C_{a_{v-1}}^{2}}{C_{a_{v}}^{2}}$, $v=2,\cdots,N-1$.
\end{cor}

Based on the above discussion, our polygamy relations of $\omega$th $(0\leq \omega \leq 2)$ power of CoA
seems to be a tight bound. 

\subsection{\textbf{Estimates of $C^{\omega}(|\psi\rangle_{AB_{1}|B_{2}\cdots B_{N-1}})$ }} \label{s:optimized bound}
The linear entropy of a state $\rho$ is defined as \cite{SF}:
\begin{equation}\label{e:chap4-ineq9}
T(\rho)=\left[1-\operatorname{Tr}\left(\rho^2\right)\right]
\end{equation}

For a bipartite state $\rho_{A B}$, $T(\rho_{A B})$ has the property \cite{ZGZG}:
\begin{equation}\label{e:chap4-ineq10}
|T(\rho_{A})-T(\rho_{B})|\leq T(\rho_{AB})\leq T(\rho_{A})+T(\rho_{B}).
\end{equation}

\bigskip

For any $N$-qubit pure state $|\psi\rangle_{AB_{1}B_{2}\cdots B_{N-1}}$, it follows from \eqref{e:chap2-ineq3} and \eqref{e:chap4-ineq9} that
\begin{equation}\label{e:chap4-ineq11}
 C^{2}(|\psi\rangle_{AB_{1}|B_{2}\cdots B_{N-1}})=2\left[1-\operatorname{Tr}\left(\rho_{AB_{1}}^2\right)\right]=2T(\rho_{AB_{1}}).
\end{equation}

\bigskip
Combining with Theorem \ref{t:chap4-thm1}, we can estimate
the range of the entropy $T(\rho)$ using information of $C(|\psi\rangle)$.
\begin{thm} \label{t:chap4-thm2}
For $0\leq \omega \leq 2$ and any $N$-qubit state $|\psi\rangle_{AB_{1}B_{2}\cdots B_{N-1}}$ $(N\geq4)$,

{\bf(1)} The lower bound for $C^{\omega}(|\psi\rangle_{AB_{1}|B_{2}\cdots B_{N-1}})=C^{\omega}(\psi)$ is as follows:
\begin{align*}
 C^{\omega}(\psi)\geq \max \left\{\left(\sum_{i=2}^{N-1}C^{2}_{AB_{i}}+C^{2}_{AB_{1}}\right)^{\frac{\omega}{2}}-\Xi_{B_{1}},\left(\sum_{i=2}^{N-1}C^{2}_{B_{1}B_{i}}+C^{2}_{AB_{1}}\right)^{\frac{\omega}{2}}-\Xi_{A}, 0\right\}
 \end{align*}

{\bf(2)} The upper bound for $C^{\omega}(\psi)$ is given by
\begin{align*}
C^{\omega}(\psi)\leq \Xi_{A}+\Xi_{B_{1}}
   \end{align*}
where  ~~ $\Xi_{j}=\sum_{k=1}^{N-1}\left[k^{\frac{\omega}{2}}-(k-1)^{\frac{\omega}{2}}\right]C_{a_{k_{j}}}^{\omega}+\frac{\omega}{2}\sum_{v=2}^{N-1}\frac{(v-1)^{2}}{v^{2}}\left[\tau_{v_{j}}^{\frac{\omega}{2}-1}-(v-1)^{\frac{\omega}{2}-1}\right]C_{a_{v_{j}}}^{\omega}$
 and $\tau_{v_{j}}=\frac{C_{a_{1_{j}}}^{2}+\cdots+C_{a_{(v-1)_{j}}}^{2}}{C_{a_{v_{j}}}^{2}}$, $v=2,\cdots,N-1$, $j=A, B_{1}$.
\end{thm}

\begin{proof} {\bf(1)}
If $C^{2}(|\psi\rangle_{A|B_{1}B_{2}\cdots B_{N-1}}) \leq C^{2}(|\psi\rangle_{B_{1}|AB_{2}\cdots B_{N-1}})$, then we have
\begin{align*}
& C^{\omega}(|\psi\rangle_{AB_{1}|B_{2}\cdots B_{N-1}})=\left(2T(\rho_{AB_{1}})\right)^{\frac{\omega}{2}}\geq |2T(\rho_{A})-2T(\rho_{B_{1}})|^{\frac{\omega}{2}}\quad (\text{by Eqs. \eqref{e:chap4-ineq10}, \eqref{e:chap4-ineq11}})\\
&=| C^{2}(|\psi\rangle_{A|B_{1}B_{2}\cdots B_{N-1}})- C^{2}(|\psi\rangle_{B_{1}|AB_{2}\cdots B_{N-1}})|^{\frac{\omega}{2}} \quad \quad(\text{by Eq. \eqref{e:chap2-ineq3}})\\
     &\geq C^{\omega}(|\psi\rangle_{B_{1}|AB_{2}\cdots B_{N-1}})- C^{\omega}(|\psi\rangle_{A|B_{1}B_{2}\cdots B_{N-1}})\quad \quad(\text{by Lemma in \cite{JF2}})\\
     &\geq \left(\sum_{i=2}^{N-1}C^{2}_{B_{1}B_{i}}+C^{2}_{AB_{1}}\right)^{\frac{\omega}{2}}- C^{\omega}(|\psi\rangle_{A|B_{1}B_{2}\cdots B_{N-1}}) \quad \quad(\text{by Eq. \eqref{e:chap1-ineq1} and $\gamma=2$ })\\
      &\geq \left(\sum_{i=2}^{N-1}C^{2}_{B_{1}B_{i}}+C^{2}_{AB_{1}}\right)^{\frac{\omega}{2}}- \Xi_{A},
   \end{align*}
where the last inequality is due to Eq. \eqref{e:chap4-ineq4}, and note that we renamed $C_{a_{i_{A}}}=C_{aAB_{i'}}$ $(i=1,\cdots,N-1)$ so that they are in descending order.

 \bigskip
Meanwhile, if $C^{2}(|\psi\rangle_{B_{1}|AB_{2}\cdots B_{N-1}}) \leq C^{2}(|\psi\rangle_{A|B_{1}B_{2}\cdots B_{N-1}})$, we then have
\begin{align*}
& C^{\omega}(|\psi\rangle_{AB_{1}|B_{2}\cdots B_{N-1}})\geq \left(\sum_{i=2}^{N-1}C^{2}_{B_{1}B_{i}}+C^{2}_{AB_{1}}\right)^{\frac{\omega}{2}}- C^{\omega}(|\psi\rangle_{B_{1}|AB_{2}\cdots B_{N-1}})
\end{align*}
By arguments similar to Theorem \ref{t:chap4-thm1}, we have
$
C^{\omega}(|\psi\rangle_{B_{1}|AB_{2}\cdots B_{N-1}}) \leq \Xi_{B_{1}}.
$

{\bf(2)} By the above discussion, we obtain
\begin{align*}
&C^{\omega}(|\psi\rangle_{AB_{1}|B_{2}\cdots B_{N-1}})=\left((2T(\rho_{AB_{1}})\right)^{\frac{\omega}{2}}\leq \left((2T(\rho_{A})+2T(\rho_{B_{1}})\right)^{\frac{\omega}{2}}\quad (\text{by Eqs. \eqref{e:chap4-ineq10}, \eqref{e:chap4-ineq11}})\\
&=\left( C^{2}(|\psi\rangle_{A|B_{1}B_{2}\cdots B_{N-1}})+ C^{2}(|\psi\rangle_{B_{1}|AB_{2}\cdots B_{N-1}})\right)^{\frac{\omega}{2}}\quad \quad(\text{by Eq. \eqref{e:chap2-ineq3}})\\
   &\leq C^{\omega}(|\psi\rangle_{A|B_{1}B_{2}\cdots B_{N-1}})+ C^{\omega}(|\psi\rangle_{B_{1}|AB_{2}\cdots B_{N-1}})\quad \quad(\text{by Lemma in \cite{JF2}})\\
  & \leq \Xi_{A}+\Xi_{B_{1}}
   \end{align*}
\end{proof}

We remark that the inequality $|x-y|^{\omega}\geq x^{\omega}-y^{\omega}$ is tight (cf. \cite[Lemma]{JF2}).

\medskip
 Combining Corollary \ref{c:chap4-cor1} with Theorem \ref{t:chap4-thm2}, we obtain superior bounds of
 the $\omega$th ($0\leq \omega \leq 2$) power of concurrence for any $N$-qubit pure states $|\psi\rangle_{AB_{1}\cdots B_{N-1}}$ under the partition $AB_{1}$ and $B_{2}\cdots B_{N-1}$.
 Now let us use an example from \cite{YCLZ} to show these
 bounds for the entropy and entanglement measure.

\begin{exmp}
Consider the following $4$-qubit generalized $W$-class state \cite{YCLZ}:
$$
|W\rangle_{AB_{1}B_{2}B_{3}}=\lambda_{1}(|1000\rangle+\lambda_{2}|0100\rangle)+\lambda_{3}|0010\rangle+\lambda_{4}|0001\rangle.
$$
 where $\sum_{i=1}^4 \lambda_i^2=1$, and $\lambda_{i}\geq 0$ for $i=1,2,3,4$.
Then  $C(|W\rangle_{AB_{1}|B_{2}B_{3}})=2  \sqrt{(\lambda_1^2+\lambda_2^2)(\lambda_3^2+\lambda_4^2)}$, $C_{AB_{1}}=C_{aAB_{1}}=2 \lambda_1\lambda_2 $, $C_{AB_{2}}=C_{aAB_{2}}=2 \lambda_1\lambda_3 $, $C_{AB_{3}}=C_{aAB_{3}}=2 \lambda_1\lambda_4 $, $C_{B_{1}B_{2}}=C_{aB_{1}B_{2}}=2 \lambda_2\lambda_3 $, and $C_{B_{1}B_{3}}=C_{aB_{1}B_{3}}=2 \lambda_2\lambda_4 $. Setting $\lambda_1=\frac{3}{4}, \lambda_2=\frac{1}{2}, \lambda_3=\frac{\sqrt{2}}{4}, \lambda_4=\frac{1}{4},$  one has $\tau_{2_{A}}=\frac{C^{2}_{aAB_{1}}}{C^{2}_{aAB_{2}}}=2, \tau_{3_{A}}=\frac{C^{2}_{aAB_{1}}+C^{2}_{aAB_{2}}}{C^{2}_{aAB_{3}}}=6.$ Similarly,  $\tau_{2_{B_{1}}}=\frac{9}{2}, \tau_{3_{B_{1}}}=11.$

{\bf(1)} The comparison of lower bound for $C^{\omega}(|W\rangle_{AB_{1}|B_{2}B_{3}})$ $(0\leq \omega \leq 2)$:

Since $\Xi_{B_1}\leq \Xi_{A}$, Theorem \ref{t:chap4-thm2} {\bf(1)} implies that our lower bound is
\begin{align*}
Z_{1}=\left(C^{2}_{AB_{1}}+C^{2}_{AB_{2}}+C^{2}_{AB_{3}}\right)^{\frac{\omega}{2}}-\Xi_{B_{1}}=\left(\frac{63}{64}\right)^{\frac{\omega}{2}}-\Xi_{B_{1}}
\ \text{or $0$},
\end{align*}
 where
\begin{align*}
\Xi_{B_{1}}&=C^{\omega}_{aAB_{1}}+\left[\left(2^{\frac{\omega}{2}}-1\right)+\frac{\omega}{8}\left(\tau_{2_{B_{1}}}^{\frac{\omega}{2}-1}-1\right)\right]C^{\omega}_{aB_{1}B_{2}}\\
&+\left[\left(3^{\frac{\omega}{2}}-2^{\frac{\omega}{2}}\right)+\frac{2\omega}{9}\left(\tau_{3_{B_{1}}}^{\frac{\omega}{2}-1}-2^{\frac{\omega}{2}-1}\right)\right]C^{\omega}_{aB_{1}B_{3}}\\
&=\left(\frac{3}{4}\right)^{\omega}+\left[\left(2^{\frac{\omega}{2}}-1\right)+\frac{\omega}{8}\left(\left(\frac{9}{2}\right)^{\frac{\omega}{2}-1}-1\right)\right]\left(\frac{\sqrt{2}}{4}\right)^{\omega}\\
&+\left[\left(3^{\frac{\omega}{2}}-2^{\frac{\omega}{2}}\right)+\frac{2\omega}{9}\left(11^{\frac{\omega}{2}-1}-2^{\frac{\omega}{2}-1}\right)\right]\left(\frac{1}{4}\right)^{\omega}.
   \end{align*}
\bigskip
 The following lower bound is given by Eq. \eqref{e:chap4-ineq5},
\begin{align*}
Z_{2}&=\left(C^{2}_{AB_{1}}+C^{2}_{AB_{2}}+C^{2}_{AB_{3}}\right)^{\frac{\omega}{2}}-C^{\omega}_{aAB_{1}}-\left(2^{\frac{\omega}{2}}-1\right)C^{\omega}_{aB_{1}B_{2}}-\left(3^{\frac{\omega}{2}}-2^{\frac{\omega}{2}}\right)C^{\omega}_{aB_{1}B_{3}}\\
&=\left(\frac{63}{64}\right)^{\frac{\omega}{2}}-\left(\frac{3}{4}\right)^{\omega}-\left(2^{\frac{\omega}{2}}-1\right)\left(\frac{\sqrt{2}}{4}\right)^{\omega}-\left(3^{\frac{\omega}{2}}-2^{\frac{\omega}{2}}\right)\left(\frac{1}{4}\right)^{\omega}.
 \end{align*}

The following lower bound is given by Eq. \eqref{e:chap4-ineq6},
\begin{align*}
Z_{3}&=\left(C^{2}_{AB_{1}}+C^{2}_{AB_{2}}+C^{2}_{AB_{3}}\right)^{\frac{\omega}{2}}-C^{\omega}_{aAB_{1}}-\left(2^{\frac{\omega}{2}}-1\right)C^{\omega}_{aB_{1}B_{2}}-\left(2^{\frac{\omega}{2}}-1\right)C^{\omega}_{aB_{1}B_{3}}\\
&=\left(\frac{63}{64}\right)^{\frac{\omega}{2}}-\left(\frac{3}{4}\right)^{\omega}-\left(2^{\frac{\omega}{2}}-1\right)\left(\left(\frac{\sqrt{2}}{4}\right)^{\omega}+\left(\frac{1}{4}\right)^{\omega}\right).
 \end{align*}

The lower bound given by Eq. \eqref{e:chap4-ineq7} is,
\begin{align*}
Z_{4}&=\left(C^{2}_{AB_{1}}+C^{2}_{AB_{2}}+C^{2}_{AB_{3}}\right)^{\frac{\omega}{2}}-C^{\omega}_{aAB_{1}}-\frac{\omega}{2}C^{\omega}_{aB_{1}B_{2}}-\frac{\omega}{2}C^{\omega}_{aB_{1}B_{3}}\\
&=\left(\frac{63}{64}\right)^{\frac{\omega}{2}}-\left(\frac{3}{4}\right)^{\omega}-\frac{\omega}{2}\left(\left(\frac{\sqrt{2}}{4}\right)^{\omega}+\left(\frac{1}{4}\right)^{\omega}\right).
 \end{align*}

\begin{figure}[H]
\caption{Comparison of Lower Bounds I}
    \centering
    \includegraphics[width=0.8\textwidth]{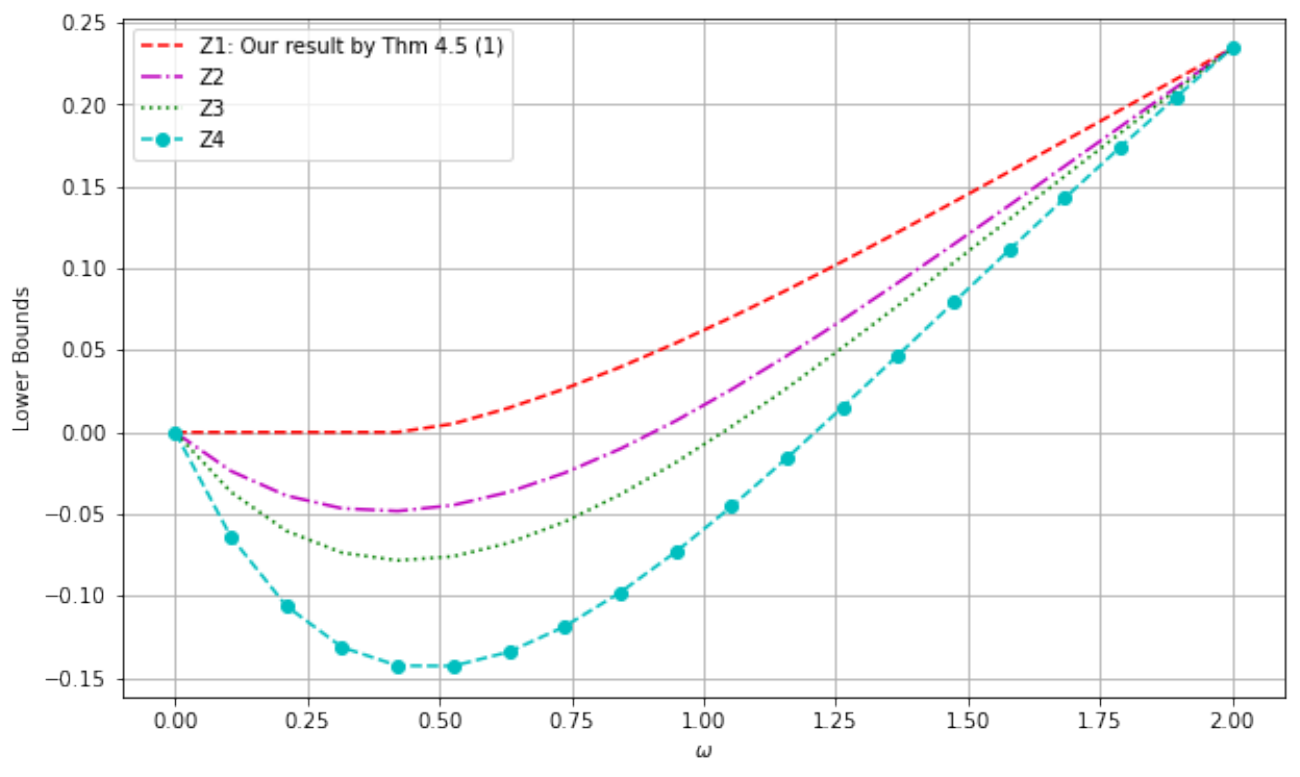}
     \captionsetup{justification=raggedright}
    \caption*{Figure 5 shows that among the lower bounds of the
$\omega$th power of $C(|W\rangle_{AB_1|B_2B_3})$
($0\leq \omega\leq 2$) the
bound $Z_1$ is the tightest one.}

\end{figure}


\begin{figure}[H]
\caption{Comparison of Lower Bounds II}
    \centering
    \includegraphics[width=0.8\textwidth]{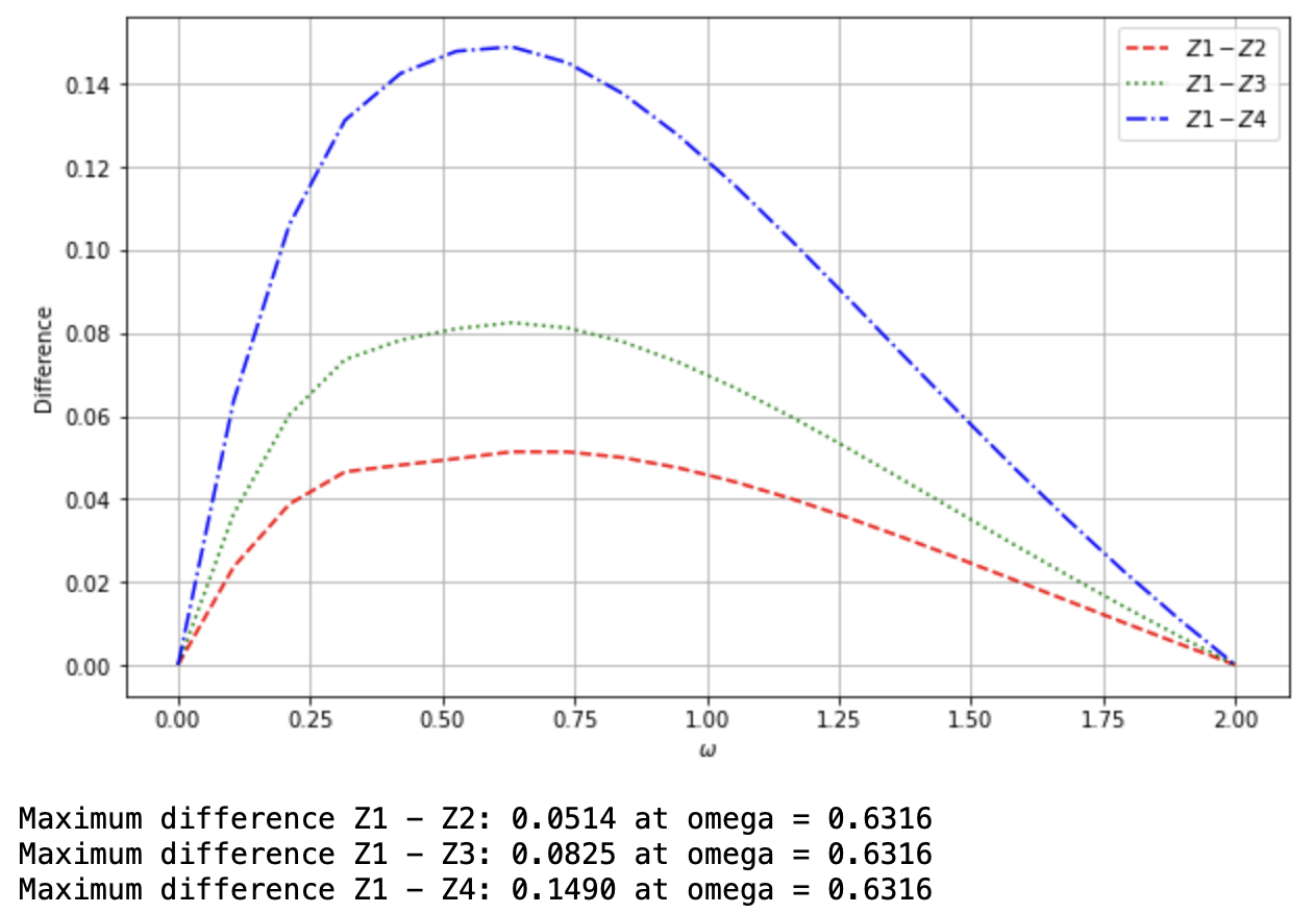}
  \caption*{Figure 6 pictures the differences and
  indicates the maxima of the differences.}  
\end{figure}

\bigskip
{\bf(2)} Comparison of upper bounds for $C^{\omega}(|W\rangle_{AB_{1}|B_{2}B_{3}})$ $(0\leq \omega \leq 2)$:

Theorem \ref{t:chap4-thm2} {\bf(2)} provides the upper bound
\begin{align*}
T_{1}=\Xi_{A}+\Xi_{B_{1}},
\end{align*}
 where $\Xi_{B_{1}}$ was calculated above, and
\begin{align*}
\Xi_{A}&=C^{\omega}_{aAB_{1}}+\left[\left(2^{\frac{\omega}{2}}-1\right)+\frac{\omega}{8}\left(\tau_{2_{A}}^{\frac{\omega}{2}-1}-1\right)\right]C^{\omega}_{aAB_{2}}\\
&+\left[\left(3^{\frac{\omega}{2}}-2^{\frac{\omega}{2}}\right)+\frac{2\omega}{9}\left(\tau_{3_{A}}^{\frac{\omega}{2}-1}-2^{\frac{\omega}{2}-1}\right)\right]C^{\omega}_{aAB_{3}}\\
&=\left(\frac{3}{4}\right)^{\omega}+\left[\left(2^{\frac{\omega}{2}}-1\right)+\frac{\omega}{8}\left(2^{\frac{\omega}{2}-1}-1\right)\right]\left(\frac{3\sqrt{2}}{8}\right)^{\omega}\\
&+\left[\left(3^{\frac{\omega}{2}}-2^{\frac{\omega}{2}}\right)+\frac{2\omega}{9}\left(6^{\frac{\omega}{2}-1}-2^{\frac{\omega}{2}-1}\right)\right]\left(\frac{3}{8}\right)^{\omega}
   \end{align*}

 The following upper bound is given by Eq. \eqref{e:chap4-ineq5},
\begin{align*}
T_{2}&=2C^{\omega}_{aAB_{1}}+\left(2^{\frac{\omega}{2}}-1\right)\left(C^{\omega}_{aAB_{2}}+C^{\omega}_{aB_{1}B_{2}}\right)+\left(3^{\frac{\omega}{2}}-2^{\frac{\omega}{2}}\right)\left(C^{\omega}_{aAB_{3}}+C^{\omega}_{aB_{1}B_{3}}\right)\\
&=2\left(\frac{3}{4}\right)^{\omega}+\left(2^{\frac{\omega}{2}}-1\right)\left(\left(\frac{3\sqrt{2}}{8}\right)^{\omega}+\left(\frac{\sqrt{2}}{4}\right)^{\omega}\right)+\left(3^{\frac{\omega}{2}}-2^{\frac{\omega}{2}}\right)\left(\left(\frac{3}{8}\right)^{\omega}+\left(\frac{1}{4}\right)^{\omega}\right).
 \end{align*}

The upper bound deduced by Eq. \eqref{e:chap4-ineq6} is,
\begin{align*}
T_{3}&=2C^{\omega}_{aAB_{1}}+\left(2^{\frac{\omega}{2}}-1\right)\left(C^{\omega}_{aAB_{2}}+C^{\omega}_{aB_{1}B_{2}}+C^{\omega}_{aAB_{3}}+C^{\omega}_{aB_{1}B_{3}}\right)\\
&=2\left(\frac{3}{4}\right)^{\omega}+\left(2^{\frac{\omega}{2}}-1\right)\left(\left(\frac{3\sqrt{2}}{8}\right)^{\omega}+\left(\frac{\sqrt{2}}{4}\right)^{\omega}+\left(\frac{3}{8}\right)^{\omega}+\left(\frac{1}{4}\right)^{\omega}\right).
 \end{align*}

The upper bound given by Eq. \eqref{e:chap4-ineq7} is,
\begin{align*}
T_{4}&=2C^{\omega}_{aAB_{1}}+\frac{\omega}{2}\left(C^{\omega}_{aAB_{2}}+C^{\omega}_{aB_{1}B_{2}}+C^{\omega}_{aAB_{3}}+C^{\omega}_{aB_{1}B_{3}}\right)\\
&=2\left(\frac{3}{4}\right)^{\omega}+\frac{\omega}{2}\left(\left(\frac{3\sqrt{2}}{8}\right)^{\omega}+\left(\frac{\sqrt{2}}{4}\right)^{\omega}+\left(\frac{3}{8}\right)^{\omega}+\left(\frac{1}{4}\right)^{\omega}\right).
 \end{align*}

 \begin{figure}[H]
 \caption{Comparison of Upper Bounds I}
    \centering
    \includegraphics[width=0.8\textwidth]{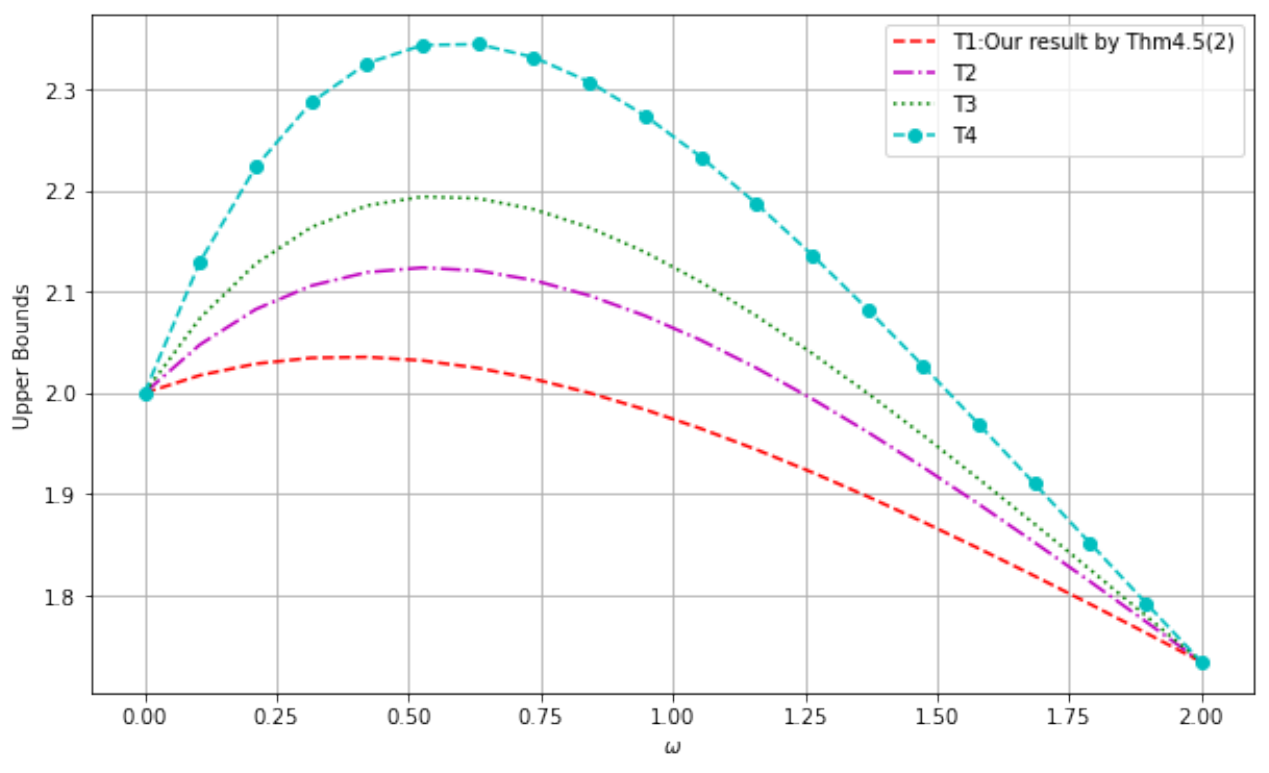}
    \captionsetup{justification=raggedright} 
    \caption*{Fig. 7 shows that among the upper bounds of the
$\omega$th power of $C(|W\rangle_{AB_1|B_2B_3})$
($0\leq \omega\leq 2$) the
bound $T_1$ is the tightest one. }
\end{figure}


\begin{figure}[H]
\caption{Comparison of Upper Bounds II}
    \centering
    \includegraphics[width=0.8\textwidth]{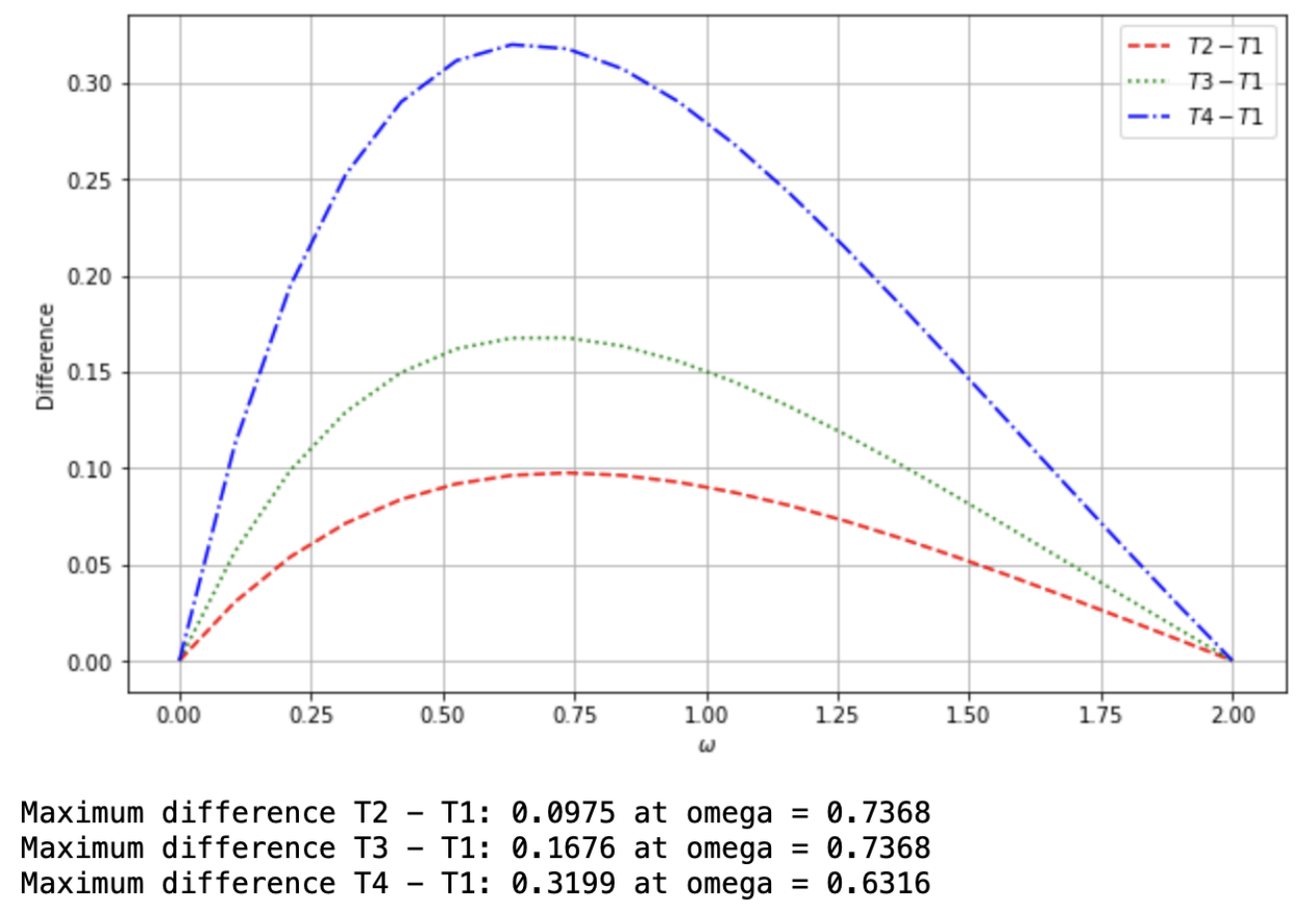}
    \captionsetup{justification=raggedright} 
    \caption*{Fig.8 shows their differences and the maxima are indicated in the description.}
\end{figure}
\end{exmp}

\section{\textbf{Conclusion}}

Various monogamy relations exist for different entanglement measures that are important in quantum information processing. Recently, we presented a family of tighter parameterized $\alpha$th-monogamy $(\alpha\geq\gamma)$ relations \cite{CJMW} based on Eq. \eqref {e:chap1-ineq1}. Therefore, there are three remaining cases that need to be discussed. Our goals in this work is to propose tighter monogamy relation for the $\alpha$th ($0\leq \alpha\leq \gamma$) power of $\mathcal{E}$ based on Eq. \eqref {e:chap1-ineq1}, as well as some good bounds for the $\beta$th ($ \beta \geq \delta$) power and $\omega$th ($0\leq \omega \leq \delta$) power of any bipartite assisted measure $\mathcal{E}_{a}$ based on Eq. \eqref {e:chap1-ineq2} in a unified manner. We discuss the monogamy and polygamy relations corresponding to these three cases respectively. It is noted that our treatment works for an arbitrary measurement.
These results are useful for exploring the entanglement theory, quantum information processing and secure quantum communication.

\bigskip

{\bf Data availability statement}. All data generated or
analyzed during this study are included in this published
article.

\bigskip

{\bf Declaration} The authors have no competing interests to declare that are relevant to the content of this article.

\bigskip

\bibliographystyle{plain}

\begin{thebibliography}{6}

\bibitem{CAF} Chen K, Albeverio S, Fei S M. Phys. Rev. Lett. 95(4), 040504 (2005).


\bibitem{HHHH} Horodecki R, Horodecki P, Horodecki M, and Horodecki K.  Rev. Mod. Phys. 81, 865 (2009).

\bibitem{DFSC} Datta A, Flammia S T,  Shaji A, Caves C M.  Phys. Rev. A 75(6), 1004 (2006).

\bibitem{ASI} Adesso G, Serafini A, Illuminati F.  Phys. Rev. A 73, 032345 (2006).

\bibitem{B} Barrett J.  Phys. Rev. A 65, 042302 (2002).

 \bibitem{CB} Cleve R, Buhrman H.  Phys. Rev. A 56, 1201 (1997).

 \bibitem{GR} Gigena N, Rossignoli R.  Phys. Rev. A 95, 062320 (2017).


\bibitem{CKW} Coffman V, Kundu J, Wootters W K. Phys. Rev. A 61, 052306 (2000).

\bibitem{OV} Osborne T J, Verstraete F.  Phys. Rev. Lett. 96(22), 220503 (2006).

 \bibitem{G1} Giorgi G L. Phys. Rev. A 84, 054301 (2011).

\bibitem{JHC} Choi J H, Kim J S. Phys. Rev. A 92(4),
    042307 (2015).

\bibitem{JLLF} Jin Z X, Li J, Li T, et al. Phys. Rev. A 97, 032336 (2018).

\bibitem{ZF1} Zhu X N, Fei S M. Phys. Rev. A 90, 024304 (2014).

\bibitem{KPSS} Kumar A, Prabhu R, Sen(De) A and Sen U.
    Phys.
Rev. A 91 012341 (2015).

\bibitem{KDS} Kim J S, Das A, Sanders B C. Phys. Rev. A, 79 012329 (2009).

\bibitem{OF} Ou Y C, Fan H. Phys. Rev. A 75(6), 062308 (2007).


\bibitem{GMS} Gour G, Meyer D A, Sanders B C. Phys. Rev. A 72, 042329 (2005).

\bibitem{GBS} Gour G, Bandyopadhay S, Sanders B C. J. Math. Phys. 48, 012108 (2007).


\bibitem{K2} Kim J S. Phys. Rev. A 97, 042332 (2018).

\bibitem{GG} Gour G, Guo Y. Quantum 2, 81 (2018).

\bibitem{G} Guo Y. Quant. Inf. Process. 17, 222 (2018).

 \bibitem{ZF2} Zhu X N, Fei S M. Phys. Rev. A 92, 062345 (2015).

\bibitem{JF1} Jin Z X, Fei S M. Quant. Inf. Proc. 16, 77 (2017).

\bibitem{YCFW} Yang L M , Chen B, Fei S M, Wang Z X. Commun. Theor. Phys. 71, 545 (2019).

\bibitem{ZJZ} Zhang M M, Jing N, Zhao H. Quant. Inf. Process. 21, 136 (2022)

 \bibitem{TZJF} Tao Y H, Zheng K, Jin Z X, Fei S M.
    Mathematics 11, 1159 (2023).

\bibitem{JFQ} Jin Z, Fei S, Qiao C.  Quant. Inf. Process. 19, 101 (2020).


\bibitem{ZLJM} Zhang X, Jing N, Liu M, Ma H T. Phys. Scr. 98, 035106 (2023).


\bibitem{CJW} Cao Y, Jing N, Wang Y L. Laser Phys. Lett. 21, 045205 (2024).

\bibitem{CJMW} Cao Y, Jing N, Misra K, Wang Y L. Quant. Inf. Process. 23, 282 (2024).


\bibitem{U} Uhlmann A. Phys. Rev. A 62, 032307 (2000).

\bibitem{RBCGM} Rungta P, Buzek V, Caves C M, Hillery M, Milburn G J. Phys. Rev. A 64, 042315 (2001).




 \bibitem{YS} Yu C S, Song H S.  Phys. Rev. A 77, 032329 (2008).

\bibitem{YL} Yang X, Luo M X. Quant. Inf. Process. 3, 20 (2021).


\bibitem{ZF3} Zhu X N, Fei S M. Quant. Inf. Proc. 18, 1-13 (2018).

\bibitem{YCLZ} Yang Y, Chen W, Li G, et al. Phys. Rev. A 97, 012336 (2019).

\bibitem{JFL} Jin Z X, Fei S M, Li-Jost X. Int. J. Theor. Phys. 58 1576C1589 (2019).

\bibitem{SF} Santos E, Ferrero M. Phys. Rev. A 62, 024101 (2000).

\bibitem{ZGZG} Zhang C J, Gong Y X, Zhang Y S, et al. Phys. Rev. A 78, 042308 (2008).

\bibitem{JF2} Jin Z X, Fei S M.  Quant. Inf. Process. 19, 23 (2020).





\end{thebibliography}

\end{document}